\documentclass[10pt,twocolumn,twoside]{IEEEtran}

\ifCLASSOPTIONcompsoc
 \usepackage[caption=false,font=normalsize,labelfont=sf,textfont=sf]{subfig}
\else
 \usepackage[caption=false,font=footnotesize]{subfig}
\usepackage{cite}
\usepackage{comment}
\usepackage{url}
\usepackage{booktabs} 
\usepackage{amsfonts,amsmath,amssymb,amsthm,bm,algorithm,algorithmic,cases,array,graphicx,caption,xcolor}
\usepackage{wrapfig}
\usepackage{cases}
\usepackage{url}
 
\theoremstyle{plain}
\newtheorem{thm}{Theorem}
\newtheorem{lem}{Lemma}
\newtheorem{prop}{Proposition}
\newtheorem{defn}{Definition}
\newtheorem{conj}{Conjecture}

\AtBeginDocument{%
  \providecommand\BibTeX{{%
    \normalfont B\kern-0.5em{\scshape i\kern-0.25em b}\kern-0.8em\TeX}}}

\ifodd 1
    
    \newcommand{\com}[1]{\textbf{\color{blue} (COMMENT: #1)}}
\else
    
    \newcommand{\com}[1]{}
\fi

\begin{document}
\title{{\em Stay} or {\em Switch}: Competitive Online Algorithms for Energy Plan Selection in Energy Markets with Retail Choice}

\author{Jianing~Zhai,
        Sid~Chi-Kin Chau,~\IEEEmembership{Member,~IEEE}
        and~Minghua~Chen,~\IEEEmembership{Senior~Member,~IEEE}
\thanks{This paper is an extension of the conference version of ACM e-Energy, 2019 \cite{eenergy}.}        

\thanks{J. Zhai is with The Chinese University of Hong Kong, Shatin, N.T., Hong Kong, (e-mail:zj117@ie.cuhk.edu.hk).}
\thanks{S. C.-K. Chau is with Australian National University, Canberra, Australia, (e-mail:sid.chau@anu.edu.au).}
\thanks{M. Chen was with The Chinese University of Hong Kong, Shatin, N.T., Hong Kong. He is now with City University of Hong Kong ,Tat Chee Avenue
Kowloon, Hong Kong, (e-mail:minghua.chen@cityu.edu.hk).}
}

\maketitle
\begin{abstract} 
Customers are now able to switch energy plans among competitive energy retail  suppliers in several countries, such as the US, the UK, Australia. However, there appears subsided participation from residential customers in retail energy markets, despite the benefits of having more affordable prices and better savings, because energy plan selection is a complex online decision-making process with a large number of options. This paper tackles the online energy plan selection problem by providing effective competitive online algorithms. An online energy plan selection problem is formulated as a metrical task system with temporally dependent switching costs. For the case of constant cancellation fee, a deterministic 3-competitive online algorithm and a randomized 2-competitive online algorithm are devised. The competitive ratios of the respective algorithms are then shown to be the best possible. Our algorithm can be extended to the case with linear cancellation fee proportional to the remaining contract duration. Through empirical evaluations using real-world household and energy plan data, our deterministic online algorithm costs no more than an additional 37\% expense compared to the optimal offline algorithm, and our randomized online algorithm further reduces this gap to 9\%. 
\end{abstract}  
\section{Introduction} 
\label{sec:intro} 

\IEEEPARstart{E}{nergy} markets have been notoriously operated by monopolies with vertically integrated providers spanning energy generation, transmission and distribution. Recently, retail choice in energy markets was introduced to enable residential, industrial and commercial customers to select retail suppliers in electricity and natural gas in several countries like the US, the UK, Australia, New Zealand, Denmark, Finland, Germany, Italy, the Netherlands and Norway \cite{zhou17retail, MK16retail}. 

This trend is driven by deregulation to promote more competitive retail energy  markets, giving customers higher transparency and more options. Liberating the retail energy  markets to competitive suppliers not only allows more affordable tariff schemes and better savings to customers but also encourages more customer-oriented and flexible services from suppliers. Furthermore, the emergence of virtual power plants \cite{WML12vpp} where household PVs and batteries are aggregated as an alternate supplier, may become a new form of energy suppliers in competitive retail energy markets. 

In competitive retail energy markets (including electricity and natural gas), there is a separation between utility providers (who are responsible for the management of energy transmission and distribution infrastructure, and the maintenance for ensuring its reliability) and energy suppliers (who generate and deliver energy to utility providers). Energy suppliers are supposed to compete in an open marketplace by providing a range of services and tariff schemes. With retail choice, customers can compare different energy plans from multiple suppliers and determine the best energy plans that suit their needs. The switching from one supplier to another can be attained conveniently via an online platform, or through third-party assistance services. 

Competitive retail energy markets have been deployed worldwide. In the US, there are over 13 states offering  retail electricity choice \cite{zhou17retail, L18retail}. In particular, there are over 109 retail electric providers in Texas offering more than 440 energy plans, including 97 of which generated from all renewable energy sources \cite{doc17texas}. In the UK, there are over 73 electricity and natural gas suppliers \cite{ofgem18suppliers}. In Australia, there are over 33 electricity and natural gas retailers and brands \cite{AEMC18retail}. 

Despite the promising goals of retail choice in energy markets, there appears subsided participation, particularly from residential customers. Declining residential participation rates and diminishing market shares of competitive retailers have been reported in recent years \cite{eia18decline, GRSCH18retail}. While there are several reasons behind the subpar customer reactions, one identified major reason is the confusion and complication of the available energy plans in retail energy markets \cite{AEMC18retail}. First, there are rather complex and confusing tariff structures by energy retailers. It is not straightforward for average consumers to comprehend the details of consumption charges and different tariff schemes. Second, savings and incentives among energy retailers are not easy to compare. There are a variety of factors to determine the actual cost of energy plans, including usage patterns and peak behavior. Some discounts are conditional on ambiguous contract terms. Without discerning the expected benefits of switching their energy plans, most customers are reluctant to participate in retail energy markets. Third, there lack of proper evaluation tools for customers to keep track of their energy usage and expenditure. Last, the increasing market complexity with a growing number of retailers and agents obscures the benefits of retail choice. For example, using real-world official datasets, we found 165 energy plans available in Buffalo, the US \cite{ny.gov}, and 434 electricity plans available in Sydney, Australia \cite{gov.au}. A large number of available energy plans cause considerable confusion and complexity to average customers. 

To bolster customer participation, several government authorities and regulators have launched websites and programs to educate customers on the benefits of retail choice in energy markets \cite{powertochoose, powertochooseNYS, energymadeeasy, ofgem18compare}. Recently, a number of start-up companies emerged to capitalize on the opportunities of retail choice in energy markets by providing assistance services and online tools to automatically determine the best energy plans for customers and to offer personalized selection advice. These assistance services and online tools are integral to the success of retail choice in energy markets by automating the confusing and complex decision-making processes of energy plan selections. In the future, household PVs, batteries, and smart appliances will optimize their usage and performance in conjunction with energy plan selection. Therefore, we anticipate the importance of proper decision-making processes for energy plan selection in energy markets with retail choice.

There are several challenging research questions arisen in the decision-making processes for energy plan selection: 
\begin{enumerate} 
	
	\item 
	{\bf Complex Tariff Structures:} There are diverse tariff structures and properties in practical energy plans. For example, the tariff schemes may have different contract periods (e.g., 6 months, one year, two years). There may be different time-of-use and peak tariffs as well as dynamic pricing depending on renewable energy sources. Also, it is common to have various administrative fees, such as connection, disconnection, setup, and cancellation fees. When there are rooftop PVs or home batteries, there will be feed-in tariffs for injecting electricity into the grid. In some countries, different appliances can be charged by different rates (e.g., boilers and heaters are charged differently). 
	
	\item 
	{\bf Uncertain Future Information:} To decide the best energy plan that may last for a long period, one has to estimate the future usage and fluctuation of energy tariffs. It is difficult to predict the uncertain information accurately to make the best decisions. For example, energy tariffs may depend on global energy markets. If there are rooftop PVs, their performance is conditional on unpredictable weather. The various sources of uncertainty complicate the decision-making processes of energy plan selection. 
	
	\item 
	{\bf Assurance of Online Decision-Making:} The energy plan selection decision-making processes are determined over time when new information is gradually revealed (e.g., the present demand and updated tariffs). A sequential decision-making process with a sequence of gradually revealed events is called {\em online decision-making}. The average customers are reluctant to switch to a new energy plan unless there is absolute assurance provided to their selection decisions. The online decision-making processes should incorporate proper bounds on the optimality of decisions as a metric of confidence for customers. 
	
\end{enumerate} 

In this paper, we shed light on the online energy plan selection problem with practical tariff structures by offering effective algorithms. Our algorithms enable automatic systems that monitor customers' energy consumption and newly available energy plans from an energy market with retail choice, and automatically recommend customers the best plans to maximize their savings. We note that a similar idea has been explored recently in practice \cite{theguardian}. 

Our results are based on {\em competitive online algorithms}. Decision-making processes with incomplete knowledge of future events are common in daily life. For decades, computer scientists and operations researchers have been studying sequential decision-making processes with a sequence of gradually revealed events and analyzing the best possible strategies regarding the incomplete knowledge of future events. For example, see \cite{BEY05online} and the references therein. Such strategies are commonly known as {\em online algorithms}. The analysis of online algorithms considering the worst-case impacts of incomplete knowledge of future events is called {\em competitive analysis}. Competitive online algorithms can assure online decision-making processes with uncertain future information. 

Our contributions are summarized in the following. 

\begin{enumerate}
    \item An online energy plan selection problem is formulated as a metrical task system with temporally dependent switching costs in Section \ref{sec:prob}.
    
    \item  For constant cancellation fee, the problem is reduced to a discrete online convex optimization problem. In Section \ref{sec:cons}, an optimal offline algorithm, a deterministic 3-competitive online algorithm, and a randomized 2-competitive online algorithm are presented for solving the energy plan selection problem. 
    
    \item Our algorithm can be extended to the case with linear cancellation fee proportional to the remaining contract duration in Section \ref{sec:dim}.

    \item Through empirical evaluations using real-world data in Section \ref{sec:empi}, our deterministic online algorithm has an average of 14.6\% in cost-saving compared to 16.2\% by optimal offline algorithm. Furthermore, our randomized online algorithm can further improve the cost-saving.
\end{enumerate}

\section{Related Work}
\label{sec:rela}

In this paper, problem \textbf{SP} lies in the domain of the Metrical Task System (MTS) problem \cite{BEY05online} and online convex optimization (OCO) problem with switching cost \cite{oco}.

For a general $n$ discrete states setting, the MTS problem is known to have a competitive ratio of $2n-1$ \cite{mts}. The energy generation scheduling in microgrids problem has been investigated in \cite{chase}, where the problem belongs to a subclass of MTS problems with convex objective function and linear switching cost. A 3-competitive online algorithm named CHASE is proposed and reaches the lower bound for any deterministic online algorithms. Remarkably, it shows that the optimal competitive ratio (termed as the price of uncertainty) for certain types of MTS problems can reduce from $2n-1$ to a constant $3$.

On the other hand, if the state space is continuous in \textbf{SP}, the problem belongs to the OCO problem with switching cost. In \cite{lcp}, an online algorithm named LCP is proved to be 3-competitive. Recently, a paper \cite{disc} shows that by proper rounding in LCP, the competitive ratio can still be maintained in the discrete state setting. Furthermore, it proves that 3 is the lower bound on the competitive ratio of deterministic online algorithms in the discrete setting. 

Moreover, \cite{disc} and \cite{disc-ext} propose a 2-competitive randomized online algorithm based on \cite{tight}. 
In this paper, we construct a uniform version out of the algorithms stated in these papers. Our new algorithm is easier to be understood and more practical to be implemented. 

Finally, we believe this is the very first study taking the temporally dependent switching cost into consideration. As a consequence, our problem is not only unique from others but also difficult to be dealt with.
\section{Problem Formulation} 
\label{sec:prob} 

In this section, we present a formal mathematical model of the energy plan selection problem. In the energy market, energy retailers offer various energy plans for households to choose from. This energy plan selection problem is stated in \textbf{EPSP}, which includes households' arbitrary demands, market time-varying prices, and plan cancellation fees. Although energy plans differ around the world, they have the same main character in choosing plans. As a typical model, electricity plans offered in the US are chosen.
Key notations are defined in Table \ref{tab:not}, and acronyms are listed in Table \ref{tab:acr}. 

\subsection{Model} 

\subsubsection{Uncertain Electricity Demands} Let the electricity demand of a household at time $t$ be $e_t$, which forms an arbitrary non-negative number sequence along timeline. Note that online algorithm design does not rely on any historical data or stochastic model of $e_t$. 

\subsubsection{Pricing Schemes} In the US, there are two major types of energy plans offered to households. Let $\mathcal{N}$ be the available plan set, and $s_t\in \mathcal{N}$ be the chosen plan at time $t$.

\begin{itemize} 
	\item {\em Variable-Rate Plan}: Variable-rate plans usually have relatively high costs with fluctuations according to the market. Switching to variable-rate plans do not incur any additional cost. Denote by $p^{s_t}_t$ the price per electricity unit for the variable-rate plan $s_t$ at time $t$. The energy cost at time $t$ is $p^{s_t}_t\cdot e_t$.
	
	\item {\em Fixed-Rate Plan}: Fixed-rate plans have less fluctuation in prices and are normally cheaper than variable-rate plans.
	However, if electricity usage has a large variation, fixed-rate plans will incur high charges. Further, the cancellation fee will be charged when terminating a fixed-rate plan. 
	
	Note that a specific tiered-pricing scheme is applied to fixed-rate plans \cite{eastcoast}. Let $B_t$ be the baseload for a household at time $t$, and $p^{s_t}_t$ be the price per electricity unit for the fixed-rate plan. A household pays $B_t\cdot p^{s_t}_t$ when electricity demand $e_t$ is between $0.9B_t$ and $1.1B_t$. Otherwise, an under-usage fee is charged at rate $U^{s_t}$ when $e_t$ is less than $0.9B_t$, or an over-usage fee is charged at rate $V^{s_t}$ when $e_t$ exceeds $1.1B_t$.
\end{itemize} 

\begin{table}[t] 
    \centering
	\caption{Key notation.} 
	\label{tab:not} 
	\begin{tabular}{@{}c@{}| p{.35\textwidth}} 
		\toprule 
		\textbf{Variable}&\textbf{Definition}\\ 
		\midrule 
		$T$ & The total number of time intervals (unit: month)\\ 
		$\sigma_t$ & The joint input at time $t$, i.e., $\sigma_t \triangleq (e_t, p_t^{s}, U^{s}, V^{s}, B_t)$\\ 
		$s_t$ & The selected plan at time $t$\\ 
		$e_t$ & The electricity demand of customer at time $t$ (kWh)\\
		$B_t$ & The base load at time $t$ (kWh)\\ 
		$p^s_t$ & The price per unit of electricity usage for plan $s$(\$ / kWh)\\ 
		$U^s$ & The under-usage charging rate for plan $s$ (\$ / kWh)\\
		$V^s$ & The under-usage charging rate for plan $s$ (\$ / kWh)\\ 
		$L^s$ & The total contract length for plan $s$ (month)\\
		$s_{[\tau,t]}$ & The chosen state sequence from time $\tau$ to time $t$, i.e., $s_{[\tau,t]}\triangleq (s_{\tau}, s_{\tau+1}, \cdots, s_{t})$\\ 
		$\beta(s_{[\tau,t]})$ & The cancellation fee at time $t$ (\$)\\ 
		$\gamma^s$ & The fixed cancellation fee for plan $s$ (\$)\\ 
		$\alpha^s$ & The cancellation fee for the residual time in the contract (\$ / month)\\
		$l_t$ & The remaining time in the contract (month)\\
		\bottomrule 
	\end{tabular} 
\end{table} 

\begin{table}[t]  
    \centering
	\caption{Acronyms for problems and algorithms.} 
	\label{tab:acr} 
	\begin{tabular}{@{}c|l@{}} 
		\toprule 
		\textbf{Acronym}&\textbf{Meaning}\\ 
		\midrule 
		\textbf{EPSP} & Energy Plan Selection Problem\\ 
		\textbf{SP} & Simplified version of \textbf{EPSP}\\ 
		\textbf{dSP} & \textbf{EPSP} with linearly decreasing switching costs\\ 
		OFA$_s$ & The optimal offline algorithm for \textbf{SP}\\
		gCHASE$_s$ & Generalized version of deterministic online \\
		& algorithm CHASE$_s$ for \textbf{SP}\\
		gCHASE$_s^r$ & Randomized version of gCHASE$_s$ for \textbf{SP}\\ 
		\bottomrule 
	\end{tabular} 
\end{table} 

Define the cost function $g_t(s_t)$ based on  input $\sigma_t \triangleq (e_t, p_t^{s}, U^{s}, V^{s}, B_t)$ and the selected plan $s_t$ at time $t$ as: 
{
\begin{numcases}{g_t(s_t) \triangleq } 
\label{equ:gt} 
\begin{aligned} 
& e_t \cdot p^{s_t}_t, \mbox{\ \ if } s_t \text{\ is a variable-rate plan}; \\ 
& e_t \cdot p^{s_t}_t +  V^{s_t} \cdot (e_t - 1.1B_t)^+  \\
& \qquad \qquad  + U^{s_t} \cdot (0.9B_t - e_t)^+, \\
& \mbox{\qquad \qquad if } s_t \text{\ is a fixed-rate plan}.
\end{aligned} 
\end{numcases} 
}

We claim our model can be extended to describe other energy markets, including Australia \cite{energymadeeasy}, as the cost function is determined only by input and the selected plan. In this paper, we focus on the US setting in \eqref{equ:gt}. 

\subsubsection{Contract Period} There are various contracts with different contract periods and restrictions available in the market. Variable-rate plans can be subscribed and withdrawn at any time, with contract period $L^{s_t} = 1$. On the other hand, retailers aim to retain consumers in a fixed-rate plan by stipulating a long contract length. Denote the contract length of a fixed-rate plan $s_t$ by $L^{s_t}$.

\subsubsection{Cancellation Fee} When a household cancels the fixed-rate plan, i.e. $s_t\neq s_{t-1}$, a cancellation fee of plan $s_{t-1}$ from the following types will incur: 
\begin{enumerate} 
	\item {\em Constant Fee}: A fixed cancellation fee $\gamma^{s_{t-1}}$ incurs when subscription time is less than the contract period.
	\item {\em Temporally Linear-dependent Fee}: a pre-specified fee $\alpha^{s_{t-1}}$ times the remaining period $l_t$ in the contract (e.g., \$10 per remaining month in the contract). 
\end{enumerate} 

We can define $\gamma$ and $\alpha$ for each plan $s$. Variable-rate plans have zero cancellation fee, i.e., $\gamma^s = \alpha^s = 0$. Constant cancellation fee plans have $\alpha^s = 0$. Temporally linear-dependent fee plans have $\gamma^s = 0$. According to the description, cancellation fee charged at time $t$ is calculated based on both $s_t$ and previous $L^{s_{t-1}}$ states. Define the cancellation fee incurred at time $t$ by $\beta(s_{[t-L^{s_{t-1}}:t]})$. We have
\begin{equation}
    \beta(s_{[t-L^{s_{t-1}}:t]}) \triangleq (\gamma^{s_{t-1}}+\alpha^{s_{t-1}}\cdot l_t)\cdot\bm{1}_{\{s_t\neq s_{t-1}\}}\cdot\bm{1}_{\{l_t\neq 0\}}
\end{equation}
where $\bm{1}_{\{\cdot\}}$ is the indicator function, and $l_t$ is the remaining period in the contract defined by
{
\begin{numcases}{ l_t \triangleq } 
\label{equ:l_t} 
\begin{aligned} 
& 0, \text{\ \ if } s_{t-1} = s_{t-2} = \cdots = s_{t-L^{s_{t-1}}}\\
& L^{s_{t-1}} + 1 - \min_{2\leq i\leq L^{s_{t-1}}}\{s_{t-i}\neq s_{t-1}\}, \text{o.w.}\\
\end{aligned} 
\end{numcases} 
}

\subsection{Problem Definition} 

Divide total time period $T$ into integer slots $\mathcal{T}\triangleq\{1, \cdots, T\}$, where each slot is one month, corresponding to the minimum switching period. We aim to find a solution $\bm{s} \triangleq (s_1, s_2, \cdots, s_T)$ to the following energy plan selection problem: 

\begin{align} 
     \textbf{(EPSP)\ \ }\min \  &  \ {\rm Cost}(\bm{s})\triangleq
      \sum_{t=1}^{T}\Big(g_t(s_t) +  \beta(s_{[t-L^{s_{t-1}}:t]}) \Big)\\
      \text{s.t. \ } & \beta(s_{[t-L^{s_{t-1}}:t]}) = \notag\\
      & \quad (\gamma^{s_{t-1}}+\alpha^{s_{t-1}}\cdot l_t)\cdot\bm{1}_{\{s_t\neq s_{t-1}\}}\cdot\bm{1}_{\{l_t\neq 0\}}\\
	\text{variables \ }  & s_t\in \mathcal{S}
	\triangleq
	\{0, 1, \cdots, n\}, t\in [1,T] \notag 
\end{align}
where initial state $s_0$ is 0 without loss of generality. Further we define $s_{-1} = s_{-2} = \cdots = s_{-L^0} = 0$. Total cost function ${\rm Cost}(\bm{s})$ consists of operational cost $g_t(s_t)$ defined in (\ref{equ:gt}), and remaining period $l_t$ defined in (\ref{equ:l_t}).

Note that this is the general problem definition, which includes all possible plans. We believe it is the first time in the literature to have a well-defined problem formulation for energy plan selection problems. However, this problem is hard to solve as switching cost is related to multiple states and is asymmetric. In the following, we focus on a simplified 2-state version with justification.

First, as variable-rate plans have zero cancellation fee, selecting the least-cost one at each time $t$ is always the optimal among all variable-rate plans. As a result, we could consider merely one variable-rate plan, which tracing the lowest variable-rate in the market. Second, after investigating most of the fixed-rate plans in the market, we conclude that their rates can be approximately characterized by a stable ranking in the long run. Therefore we can focus on the one with the cheapest average rate. As a result, we are able to reduce the large number of plans to two plans by properly preprocessing all plans. Denote the state of the fixed-rate plan by 0 and variable-rate plan by 1. We simplify \textbf{EPSP} as follows:

\begin{subequations} 
	\begin{align} 
	\textbf{(EPSP-2)\ \ } \min \  &  \ {\rm Cost}(\bm{s})\triangleq \sum_{t=1}^{T}\Big(g_t(s_t)+\hat{\beta_t}\cdot (s_t-s_{t-1})^+\Big) \label{equ:ecep-min}\\ 
	\text{subject to \ } &  \hat{\beta_t} = \bm{1}_{\{\sum_{t-\textsf{L}+1}^t s_\tau > 0\}} \cdot \beta_t, \label{equ:ecep-L}\\ 
	\text{variables \ } & s_t\in \mathcal{N}\triangleq\{0, 1\}, t\in [1,T], \notag 
	\end{align} 
\end{subequations} 
where $(x)^+ = \max\{x,0\}$. Constraint  (\ref{equ:ecep-L}) captures zero cancellation fee when canceling a fixed-rate plan by the end of contract length.

\textbf{Remarks:} Although our problem formulation is similar to the online optimization problems in LCP \cite{lcp}, CHASE \cite{chase} and their extended version \cite{oco, disc, rchase}, the fundamental difference in constraint (\ref{equ:ecep-L}) makes our problem more challenging. Note that the cancellation fee depends on the last $L^{s_t}$ states at time $t$. Hence, we cannot follow the  sub-problem division as that of prior studies. Therefore, our problem is harder and requires non-trivial treatments.
\section{Constant Switching Cost}
\label{sec:cons}

In this section, we focus on the essential part of \textbf{EPSP} by eliminating some constraints and  studying a simplified version. A survey  \cite{powertochooseNYS} shows that most retailers who offer fixed-rate plans do not set a minimum contract period. Further, the longer the maximum contract period is, the lower price the contract offers. Therefore, it is reasonable to assume  $L^{s}$ being very large compared to $T$, and drop constraints in (\ref{equ:ecep-L}). Besides, a large proportion of fixed-rate plans guarantee a constant cancellation fee. Under such observation, we reformulate \textbf{EPSP} to a basic problem as follows:
\begin{align}
\textbf{(SP)\qquad }  \min \ \ & {\rm Cost}(\bm{s})\triangleq \sum_{t=1}^{T}\Big(g_t(s_t)+\beta\cdot (s_t-s_{t-1})^+\Big) \label{equ:sp}\\
\text{variables \ } & s_t\in \mathcal{N}. \notag
\end{align}

The following proposition shows that  \textbf{SP} belongs to the classical online decision problem {\em Metrical Task System} (MTS)  \cite{BEY05online}, where \textbf{P2} is a typical MTS problem. 

\begin{prop}
	\label{thm:sp}
	The following two problems are equivalent under the boundary condition of $s_0=s_{T+1}=0$, $g_{T+1}(0)=0$:
	\begin{align}
	\textup{\textbf{(P1)\qquad }}  \min \ \  & 
	\sum ^{T+1}_{t=1}\Big(g_{t}(s_{t}) 
	+\beta (s_{t}-s_{t-1})^+\Big)
	\label{prob1}
	\\
	\textup{subject to} & \quad s_{t} \in \mathcal{S} \notag
	\end{align}
	\begin{align}
	\textup{\textbf{(P2)\qquad }}  \min \ \  & 
	\sum ^{T+1}_{t=1}\Big(g_{t}(s_{t}) 
	+\frac{\beta}{2}\left|s_{t}-s_{t-1}\right|\Big)
	\label{prob2}
	\\
	\textup{subject to} & \quad s_{t} \in \mathcal{S} \notag
	\end{align}
\end{prop}
See Appendix \ref{pro:sp} for the full proof.

\subsection{Offline Optimal Algorithm}
\label{sec-ofa}

Denote the input of problem \textbf{SP} as $\bm{\sigma} = (\sigma_t)_{t=1}^T$. In the offline setting, \textbf{SP} can be optimally solved since $\bm{\sigma}$ is given before time starts. An efficient offline optimal algorithm (OFA$_s$) with low space and computational complexity is obtained in this section. By discussing the behavior of OFA$_s$, we motivate our design towards competitive online algorithms for the next section.
Our idea comes from the theoretical framework in \cite{chase}, and we follow notations therein.

\begin{defn}
	\label{def:delta}
	Define the one-timeslot cost difference by
	\begin{equation}
	\delta(t)\triangleq g_t(0)-g_t(1). \label{equ-delta}
	\end{equation}
\end{defn}
If $\delta(t)$ is positive, the cost of using the variable-rate plan ($s_t$ is 1) is smaller than that of using the fixed-rate plan. Therefore, it urges us to move $s_t$ to 1 for a lower cost. Similarly, a negative  $\delta(t)$ prompts us to let $s_t$ be 0.

However, we shall not decide $s_t$ only based on $\delta(t)$. The potential frequent change in $s_t$ will bring in a high cancellation fee. As an alternative, we aim to find a balance between staying within one plan and changing to a lower cost plan by calculating the cumulative cost difference between $g_t(0)$ and $g_t(1)$ throughout a time period.

\begin{defn}
	Define the cumulative cost difference by
	\begin{equation}
	\Delta(t)\triangleq \Big(\Delta(t-1)+\delta(t)\Big)_{-\beta}^0 \ , \label{equ-Delta}
	\end{equation}
	where $(x)_a^b\triangleq\min\{b,\max\{x,a\}\}$, with the initial value  $\Delta(0) = -\beta$.
\end{defn}

When $\Delta(t)$ increases from $-\beta$ to 0 in a time interval, the cumulative cost of staying at state 0 is no less than that of changing state from 0 to 1. Therefore, $\text{OFA}_s$ should let the state be 1 throughout this time interval. Analogously, $\text{OFA}_s$ should keep the state as 0 when $\Delta(t)$ decreases from 0 to $-\beta$. In the rest time intervals, the change in $\Delta(t)$ is less than $\beta$. The state shall not alter in each interval as it will only bring in more cost. We name this algorithm as $\text{OFA}_s$ given in Algorithm 1, and an example is shown in Fig. \ref{fig:rCHASE_prob}.

\begin{thm}
	\label{thm:ofa}
	 $\text{OFA}_s$ (Algorithm 1) is an optimal offline algorithm for problem \textbf{SP}.
\end{thm}

\begin{IEEEproof} (Sketch)
	To show $\text{OFA}_s$ has the minimum cost for any given input, we divide time $T$ into segments based on solutions of $\text{OFA}_s$. Then we show that in each segment, no other solution can obtain a lower cost.
	
	See Appendix \ref{pro:ofa} for the full proof.
\end{IEEEproof}

\begin{thm}
\label{thm:tscom}
	Both the running time and space requirement of $\text{OFA}_s$ for problem \textbf{SP} are $\mathcal{O}(T)$.
\end{thm}

\begin{IEEEproof}
	First, $\Delta(t)$ is computed in $\mathcal{O}(1)$ time at each time slot $t$. Thus, $\Delta(t)$ requires at most $\mathcal{O}(T)$ for computation and storage in $T$.
	Running this algorithm takes additionally $\mathcal{O}(T)$ time and space resources based on the structure of $\text{OFA}_s$.
	In total, only $\mathcal{O}(T)$ is needed for both time and space.
\end{IEEEproof}

Further, it is evident that no algorithm can beat $\mathcal{O}(T)$ in either time or space complexity. Therefore, $\text{OFA}_s$ is an optimal offline algorithm for problem \textbf{SP} due to the length of input and solution sequence.
\subsection{Competitive Online Algorithms}
\label{sec-chase}

We present two types of online algorithms in this section. First, there is a deterministic online algorithm that produces invariable solutions for the same input sequence. Second, we show a randomized online algorithm which generates a probabilistic ensemble of solutions over time.

\subsubsection{Deterministic Online Algorithm}\hspace*{\fill} 

To begin with, we formally define the terms of the online algorithm and the competitive ratio \cite{BEY05online}.
Let ${\rm Opt}(\bm{\sigma})$ be the optimal offline cost for input $\bm{\sigma}$. An {\em online algorithm} ${\mathcal A}$ generates output $s_t$ at time $t$ based on  $(\sigma_\tau)_{\tau=1}^t$ only. If the solution of $\mathcal{A}$ is fixed for the same $\bm{\sigma}$, we call $\mathcal{A}$ a deterministic  online algorithm. Algorithm $\mathcal{A}$ is  \textit{$c$-competitive} if 
\begin{equation}
\label{equ:cr}
	{\rm Cost}_{\mathcal{A}}(s_0,\bm{\sigma})\leq c\cdot {\rm Opt}(s_0,\bm{\sigma})+\gamma(s_0),\forall \mbox{\boldmath$\sigma$},
\end{equation}
where $s_0$ is the initial state,  $\gamma(s_0)$ is a constant only related to $s_0$, and ${\rm Cost}_{\mathcal{A}}(\bm{\sigma})$ is the total cost of ${\mathcal A}$ for input $\bm{\sigma}$.
The smallest $c$ satisfying (\ref{equ:cr}) is the \textit{competitive ratio} of ${\mathcal A}$. This section aims to find a deterministic ${\mathcal A}$ with a minimal $c$ for problem \textbf{SP}.

Following the idea of $\text{OFA}_s$, we investigate how the state should change based on $\Delta(t)$. As $\Delta(t)$ is not revealed until the arrival of $t$, we do not change the state unless it is inevitable that staying in the current state has already incurred more cost than changing to another state. Namely, $s_t$ is set to 1 by the time $\Delta(t)$ increases from $-\beta$ to 0, and set to 0 when $\Delta(t)$ views a drop from 0 to $-\beta$. In other cases, $s_t$ does not change. The algorithm is formally given in Algorithm 2 and shown in Fig. \ref{fig:rCHASE_prob} by an example. The algorithm is named $\text{gCHASE}_s$ since it appears to  ``chase'' $\text{OFA}_s$ in a literal sense, where the behavior of $\text{OFA}_s$ is replicated in an online fashion. The only behavior difference between two algorithms is during time period $T^{u}$ and $T^{d}$, where $\Delta(t)$ increases from $-\beta$ to 0 and decreases from 0 to $-\beta$, respectively. 

\begin{thm}
	\label{thm:det}
	The competitive ratio of $\text{gCHASE}_s$ (Algorithm 2) for problem \textbf{SP} is 3.
\end{thm}

\begin{IEEEproof} (Sketch)
	Based on the value of $\Delta(t)$, we divide time $T$ into intervals and sort them by critical segments as in \cite{chase}. After analyzing the costs of $\text{gCHASE}_s$ and $\text{OFA}_s$ in each time interval and calculating their overall costs, it can be shown that the cost ratio upper bound of $\text{gCHASE}_s$ over $\text{OFA}_s$ is 3.
	
	See Appendix \ref{pro:det} for the full proof.
\end{IEEEproof}

\textbf{Remark:} Our algorithm $\text{gCHASE}_s$ is a generalized version of algorithm CHASE in \cite{chase}, and it adopts notations therein analogously. In \cite{chase,oco} and \cite{disc}, 3 is shown to be the lower bound on the competitive ratio of any deterministic online algorithm for problem \textbf{SP}. Further, using the same method as that in Theorem \ref{thm:tscom}, $\text{gCHASE}_s$ can be proved to have both time and space complexity of $\mathcal{O}(T)$. Hence, $\text{gCHASE}_s$ is one of the best deterministic online algorithms in terms of competitive ratio and computational complexity.

\begin{figure}[h]
    \includegraphics[width=0.49\textwidth]{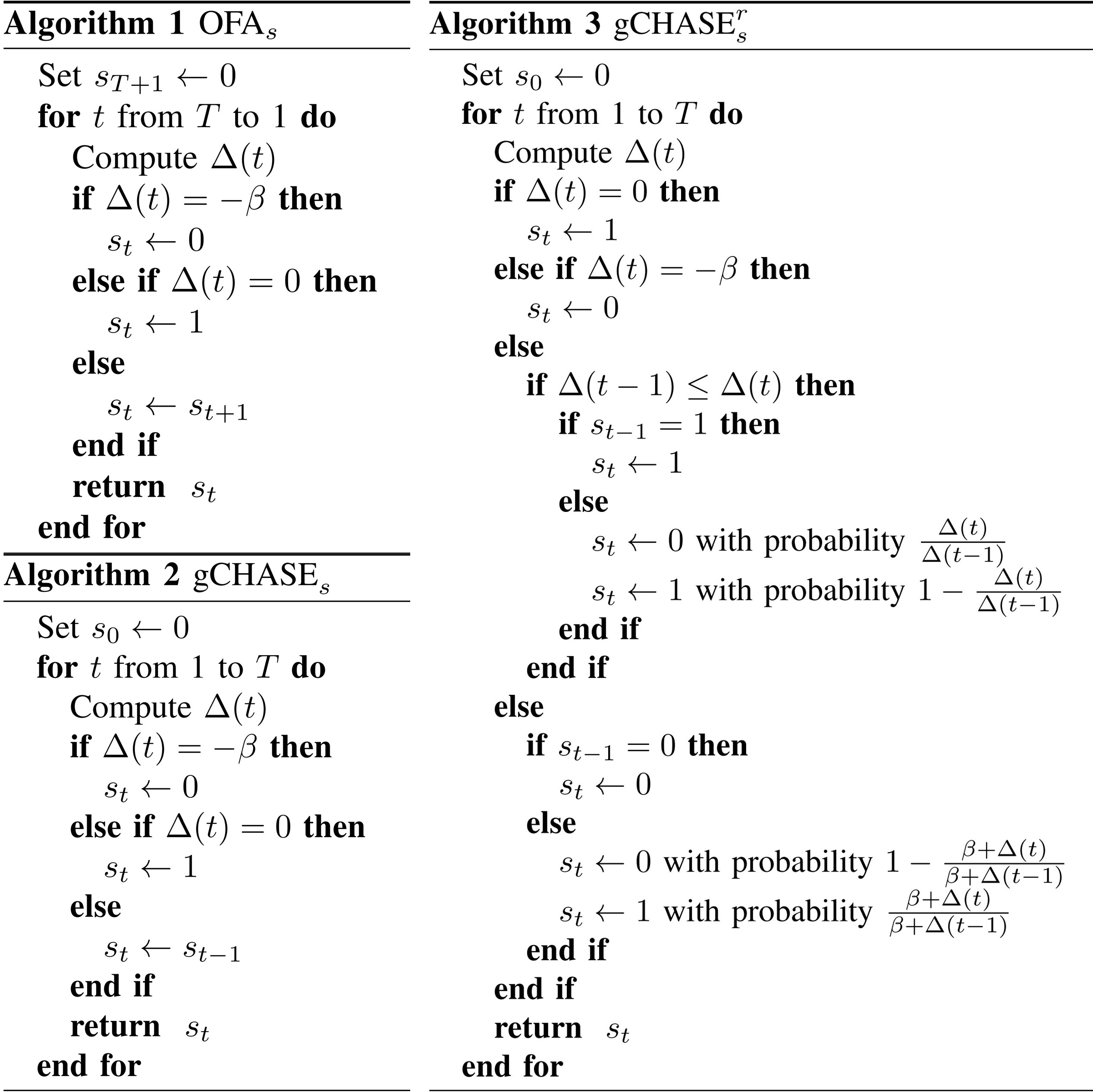}
    \label{fig:alg_all}
\end{figure}

\subsubsection{Randomized Online Algorithm}\hspace*{\fill} 

By adding randomization to an online algorithm ${\mathcal A}$, we let it make decisions probabilistically in each execution. Define the { \em expected} competitive ratio of randomized online algorithm ${\mathcal A}$ by the smallest constant $c$ satisfying
\begin{equation}
	{\mathbb E}[{\rm Cost}_{\mathcal{A}}(s_0,\bm{\sigma})]\leq c\cdot {\rm Opt}(s_0,\bm{\sigma})+\gamma(s_0),\forall \mbox{\boldmath$\sigma$},
\end{equation}
where ${\mathbb E}[ \cdot ]$ is the expectation over all random solutions. Next, we devise a competitive randomized online algorithm for problem \textbf{SP}.


\begin{figure}[h]
    \centering
    \includegraphics[width=0.48\textwidth]{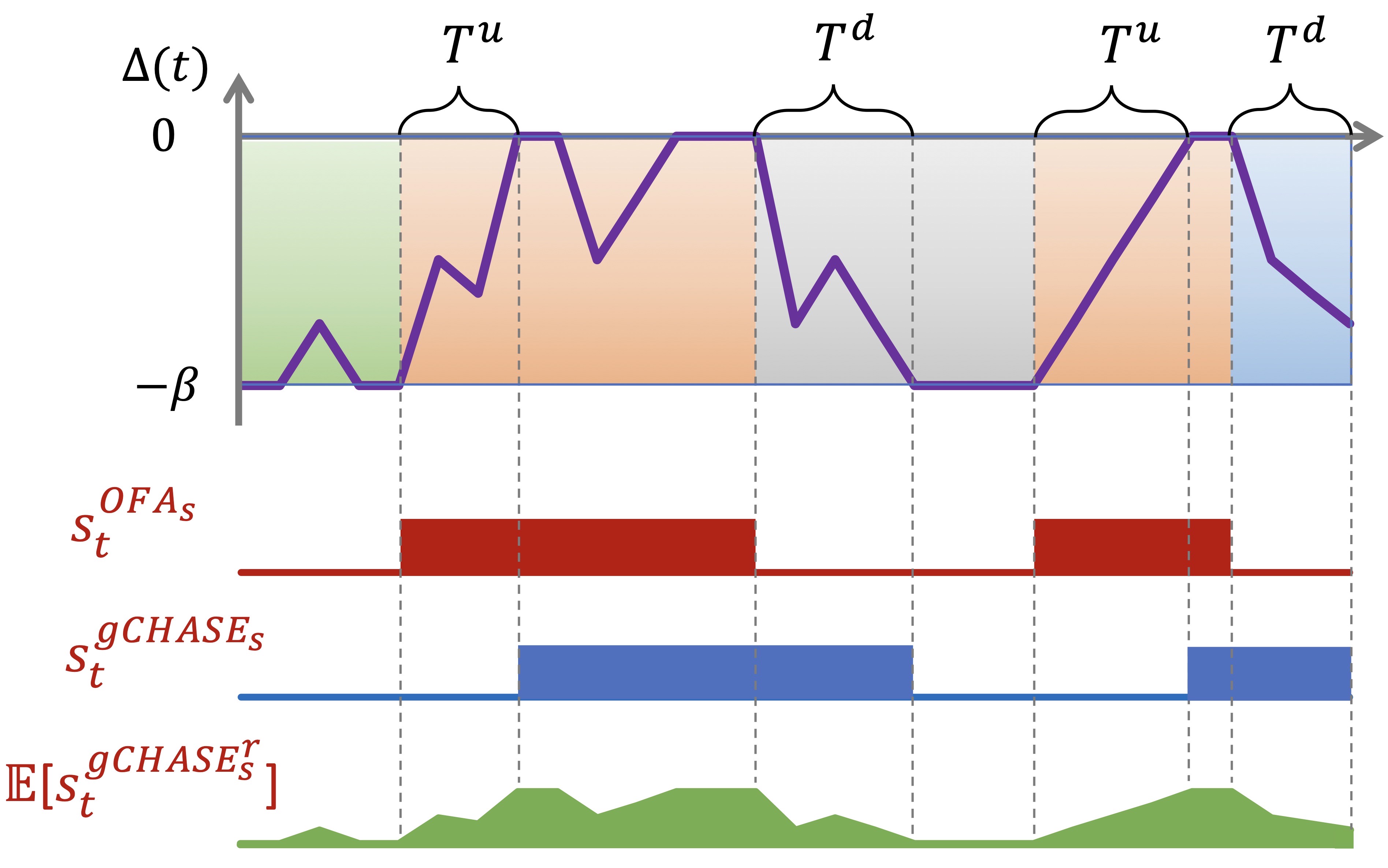}
    \caption{An example of $\Delta(t)$, $\text{OFA}_s$, $\text{gCHASE}_s$, and $\text{gCHASE}_s^r$. The solutions of running Algorithm 1, 2, and 3 are shown in the last three rows, where $s_t^{\mathcal A}$ denotes the state at time $t$ in algorithm ${\mathcal A}$, and Pr$[\cdot]$ is the probability notation.}
    \label{fig:rCHASE_prob}
\end{figure}

Recall that in $\text{gCHASE}_s$, the state does not change unless the cumulative cost difference $\Delta(t)$ already reaches 0 or $-\beta$, which makes $\text{gCHASE}_s$ a conservative online algorithm. Here we introduce randomization in order to change the state earlier. Intuitively, the nearer $\Delta(t)$ to $-\beta$ or 0 is, the higher probability state $s_t$ to 0 or 1 should take. Ideally, we want the trend of state following $\Delta(t)$. Due to the troublesome non-linear component $(s_t-s_{t-1})^+$ in \textbf{SP}, it is difficult to obtain the expected cost ${\mathbb E}[{\rm Cost}(\bm{s})]$. However, based on Lemma \ref{lem-ran}, we show that certain types of randomized online algorithms can tackle the issue by changing the non-linear term to a linear term.

\begin{lem}
\label{lem-ran}
If a randomized online algorithm satisfies 
\begin{equation}
\label{equ-lem}
    \text{Pr}[s_t>s_{t-1}]\cdot \text{Pr}[s_t<s_{t-1}]=0, \forall t\in [1,T],
\end{equation}
then its solution vector $\bm{s}$ for problem \textbf{SP} satisfies ${\mathbb E}[{\rm Cost}(\bm{s})] = {\rm Cost}({\mathbb E}[\bm{s}])$.
\end{lem}

\begin{IEEEproof}
    Denote $\mathbb{E}[s_t]$ by $\overline{s_t}$. We prove the lemma by showing two parts, $\mathbb{E}[g_t(s_t)] = g_t(\overline{s_t})$ and $\mathbb{E}[(s_t - s_{t-1})^+] = (\overline{s_t}-\overline{s_{t-1}})^+$.  Since $s_t$ is an integer, we first define $g_t(\overline{s_t})$ by interpolating between $g_t(0)$ and $g_t(1)$. Hence, function $g_t(\cdot)$ is linear, and the first part can be proved. For the second part, we prove by direct calculation. If $\text{Pr}[s_t>s_{t-1}] = 0$, function $(\cdot)^+$ always takes positive value. We have $\mathbb{E}[s_t - s_{t-1}] = \overline{s_t}-\overline{s_{t-1}}$. If $\text{Pr}[s_t<s_{t-1}] = 0$, function $(\cdot)^+$ always take value 0. We have $\mathbb{E}[(s_t - s_{t-1})^+] = (\overline{s_t}-\overline{s_{t-1}})^+ = 0$.
    See Appendix \ref{pro:lem-ran} for the complete proof.
\end{IEEEproof}

\textbf{Remark:} 1). Randomized algorithms which satisfy equation (\ref{equ-lem}) has the following property: at each time $t$, the state can only be no greater than the previous state, or no smaller than the previous state. We claim this property is reasonable in designing a randomized online algorithm. As the chosen state should be within a specified range, it is rational to set the range between the previous state and a possible optimal state. Otherwise, if the algorithm let probability spread over the whole state space without considering the previous chosen state, it is unlikely to guarantee good performance. 2). Although there are only two states in problem \textbf{SP}, result in Lemma \ref{lem-ran} holds in multi-state space if function $g_t(\cdot)$ is linear or $g_t(\cdot)$ is piece-wise linear and $s_t$ is only chosen from two consecutive integer numbers. The lemma suggests that one can evaluate the performance of the randomized online algorithms using the expected state at each time slot. 3). As ${\mathbb E}[\bm{s}]$ lies in the continuous state space, the lemma also suggests a way to derive discrete state algorithms based on continuous state algorithms, while maintaining the competitive ratio guarantee. It is the first time we know in the literature to provide a general sufficient condition on designing randomized online algorithms in discrete state space, and we suspect that the optimal randomized online algorithms should all satisfy equation (\ref{equ-lem}) for similar problems.

In the algorithm $\text{gCHASE}_s^r$ shown in Algorithm 3, when $\Delta(t)\leq \Delta(t)$, $\text{Pr}[s_t>s_{t-1}] = 0$, and when $\Delta(t)> \Delta(t)$, $\text{Pr}[s_t<s_{t-1}] = 0$. Therefore, the algorithm satisfies equation (\ref{equ-lem}). Based on the basic probability calculation, we can derive the probability of $s_t$ being 1 equals $\frac{\beta + \Delta(t)}{\beta}$.
We demonstrate the expected state behavior by an example given in Fig. \ref{fig:rCHASE_prob}. Comparing to the conservative algorithm $\text{gCHASE}_s$ the randomized online algorithm changes state more adaptively while incurring lower loss than the deterministic online algorithm.

In real-world, randomized algorithms do not necessarily entail random decisions of a single customer. When we consider an ensemble of a large number of customers using an automatic energy plan recommendation system, each customer can be given a deterministic decision rule drawn from a probabilistic ensemble of decision rules. In the end, the expected cost of a customer can be computed by the expected cost of a randomized algorithm. 


\begin{figure}[h]
    \centering
    \includegraphics[width=0.35\textwidth]{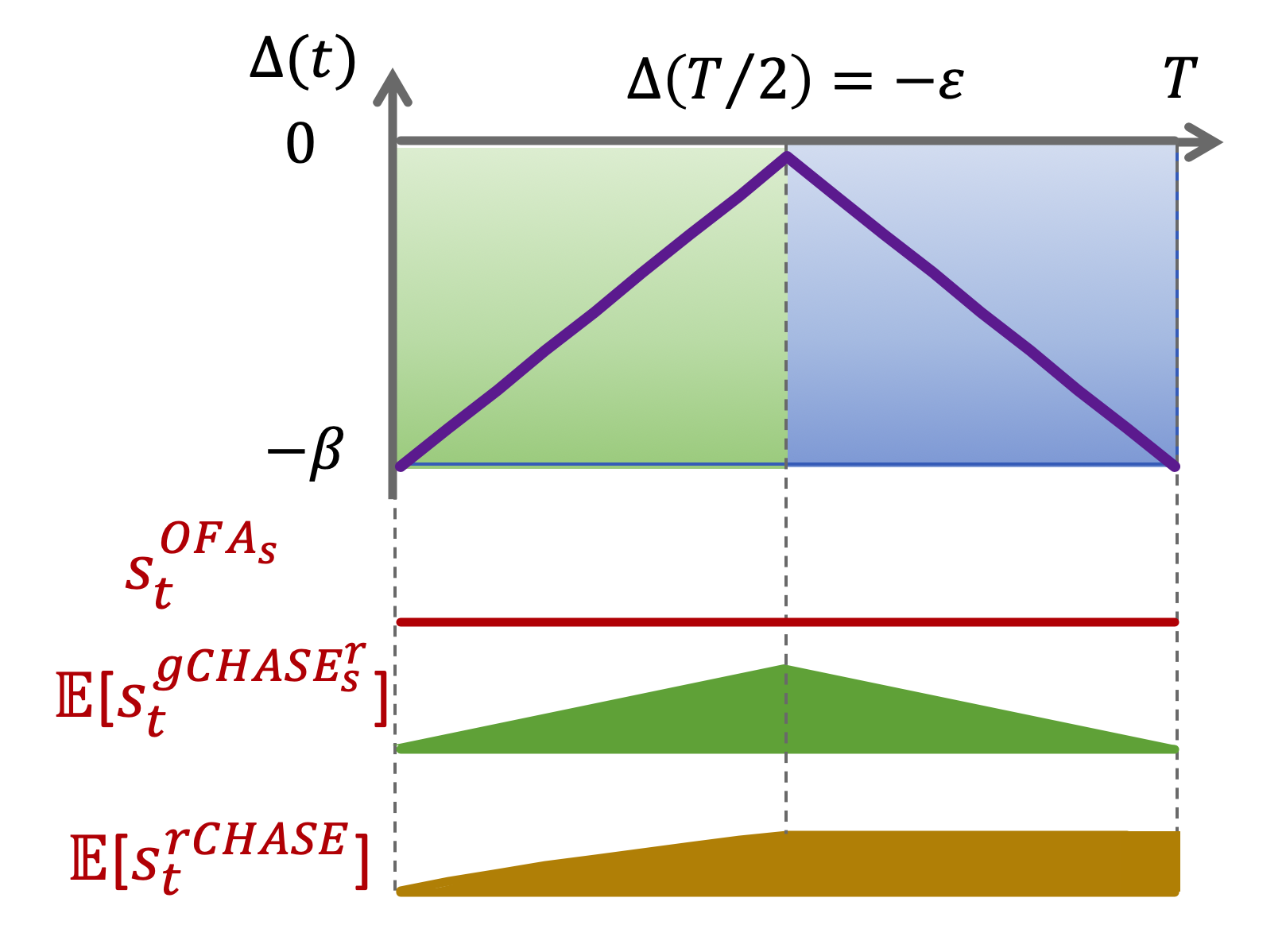}
    \caption{A worst-case $\Delta(t)$ of $\text{gCHASE}^r_s$ and rCHASE\cite{rchase}. The solutions of running Algorithm 1, 3, and rCHASE are shown in the last three rows, where $\mathbb{E}[\cdot]$ is the expectation notation.}
    \label{fig:rCHASE_worst}
\end{figure}

\textbf{Note:} People familiar with the former randomized online algorithm rCHASE in \cite{rchase} may question why $\text{gCHASE}_s^r$ guarantees a smaller competitive ratio. First, rCHASE also satisfies equation (\ref{equ-lem}). An explicit analysis is in Appendix \ref{pro:chase}. Hence it is sufficient to consider the expected state cost. We demonstrate a worst-case example in Fig. \ref{fig:rCHASE_worst}. Such a situation might exist since $\Delta(t)$ is only dependent on input. In the example, the cost of OFA$_s$ is at least $(\beta-\epsilon)$ due to the non-negative $g_t(\cdot)$. Then the cost of $\text{gCHASE}_s^r$ is $2(\beta-\epsilon)$ and the cost of rCHASE is $3(\beta-\epsilon)$. The critical feature of changing the expected value of $s_t$ simultaneously with $\Delta(t)$ lets $\text{gCHASE}_s^r$ having a lower cost.

\begin{thm}
	\label{thm:gchaser}
	The {expected} competitive ratio of $\text{gCHASE}_s^r$ (Algorithm 3) for problem \textbf{SP} is 2.
\end{thm}

\begin{IEEEproof} (Sketch)
	By relaxing the discrete states to the continuous states setting, we show that the expected cost of $\text{gCHASE}_s^r$ is equal to a continuous version $\text{cCHASE}_s^r$. Then we show that $\text{cCHASE}_s^r$ has a competitive ratio of 2 in the continuous setting, where the optimal offline algorithm is $\text{cOPT}_{s}$. Lastly, it can be verified that the optimal cost in the discrete setting is an upper bound to the continuous one. The logical is to show
        \begin{align}
    \mathbb{E}[{\rm Cost}_{\text{gCHASE}_s^r}] = {\rm Cost}_{\text{cCHASE}_{s}^{r}} \leq 2{\rm Cost}_{\text{cOPT}_{s}} \leq 2{\rm Cost}_{\text{OFA}_s}.
\end{align}

	See Appendix \ref{pro:ran} for the complete proof.
\end{IEEEproof}

\textbf{Remark:} 1). We improve the randomized algorithm in \cite{rchase} by taking the momentary change of $\Delta(t)$ into consideration. 2). The original ideas of $\text{gCHASE}_s^r$ appeared in \cite{oco} and \cite{disc}. However, their algorithms are based on the continuously updated probability distribution, which is much complicated in operation. 3). In \cite{disc-ext}, Theorem 8 proves that 2 is the lower bound on expected competitive ratio of any randomized online algorithm for problem \textbf{SP}. Therefore $\text{gCHASE}_s^r$ is one of the best randomized algorithms in terms of competitive ratio.
\section{Linearly Decreasing Switching Cost} 
\label{sec:dim}

In this section, we add back constraint (\ref{equ:ecep-L}) in \textbf{EPSP} to study the linearly decreasing cancellation fee setting, in which  $\beta_t$ is proportional to the length of remaining time in a fixed-rate plan. Due to the change of $\beta_t$, we propose a suitable problem setting and discuss its relationship with  \textbf{SP}.

First, the total time period $[0, T+1]$ is divided into consecutive time segments, where states within each segment is the same:
\begin{align}
[T_0, T_1 - 1], \cdots, [T_n, T_{n+1} - 1], \cdots, [T_{2n}, T_{2n+1} - 1],
\label{equ:segment}
\end{align}
where $T_0 = 0$, $T_{2n+1} - 1 = T + 1$, and
\begin{subequations}
	\begin{numcases}{s_t = }
	0, \quad t\in [T_{2i}, T_{2i+1} - 1], \quad \forall i\in [0,n], \label{dsp-st1}\\
	1, \quad t\in [T_{2i-1}, T_{2i} - 1], \quad \forall i\in [1,n]. \label{dsp-st2}
	\end{numcases} 
\end{subequations}

We assume that a household's each fixed-rate plan subscription period $[T_{2i}, T_{2i+1} - 1]$ does not exceed contract length {\sffamily L}. If it is not the case, we divide the period into several segments with 0 time length variable-rate subscription in between. As a result, we rewrite the problem below.
\begin{subequations}
\begin{align}
 \textbf{dSP:} \min \ {\rm Cost}(\bm{s})\triangleq & \sum_{t=1}^{T} g_t(s_t) + \sum_{i = 0}^{n}\Big(\alpha\cdot [\text{\sffamily L} - (T_{2i + 1} - T_{2i})]\Big)\\
\text{subject to \ } & T_{2i + 1} - T_{2i} \in [0, \text{\sffamily L}] , \forall i \in [0,n], \label{equ:length}\\
\text{variables \ } & s_t \text{\ satisfies\  (\ref{dsp-st1}) and  (\ref{dsp-st2})},\\
& 2n \in [0, T],
\end{align}
\end{subequations}
where $\alpha\cdot [\text{\sffamily L} - (T_{2i + 1} - T_{2i})]$ is the cancellation fee when switching from the fixed-rate plan to the variable-rate plan. Constraint (\ref{equ:length}) captures the maximum contract length {\sffamily L} of the fixed-rate plan.

We use Definition \ref{def:delta} for $\delta(t)$ as in problem \textbf{SP},  but define a new cumulative cost difference $\hat{\Delta}(t)$ as follows.
 
\begin{defn}
	\label{def:dDelta}
	Define the cumulative cost difference for \textbf{dSP} by
	\begin{equation}
	\hat{\Delta}(t)\triangleq \Big(\hat{\Delta}(t-1) + \delta(t) - \alpha \Big)_{-\beta}^0 \ , \label{equ-dDelta}
	\end{equation}
	where $\beta = \alpha\cdot\text{\sffamily L}$ and $\hat{\Delta}(0) = -\beta$.
\end{defn}

Intuitively, the total cancellation fee $\alpha \cdot \text{\sffamily L}$ can be divided into {\sffamily L} time slots. The fixed-rate plan tends to suffer $\alpha$ more than the variable-rate plan in each time slot because of the potential cancellation fee.
As the structure of \textbf{dSP} and \textbf{SP} are much alike, we use $\text{gCHASE}_s$ as a online heuristic algorithm, where ${\Delta(t)}$ is replaced by $\hat{\Delta}(t)$, to solve \textbf{dSP}. In Section~\ref{sec:empi}, we show it has a decent performance by empirical evaluations. We further conjecture the online heuristic algorithm $\text{gCHASE}_s$ has a competitive ratio of $(3 + \frac{1}{\text{\sffamily L}-1})$ for \textbf{dSP}. A longer contract period {\sffamily L} may lead the heuristic online algorithm to be more competitive. See Appendix \ref{pro:conj} for more discussion.

Note that OFA$_s$ cannot be served as the optimal offline algorithm for \textbf{dSP}, because the backward recurrence relation does not hold here. However, dynamic programming can always be used to construct offline algorithms as in \cite{chase, disc}. Although dynamic programming may incur a high time and space complexity, it is enough for us to guarantee optimal results. Furthermore, we anticipate that the  gCHASE$_s^r$ should also be competitive by inserting $\hat{\Delta}(t)$.

\section{Empirical Evaluations}
\label{sec:empi}

By using real-world traces, we evaluate our proposed algorithms under different settings in this section. Our objectives are threefold: ({\romannumeral 1}) comparing our proposed deterministic and randomized online algorithms by comparing their cost ratios against the optimal offline algorithm in various scenarios, ({\romannumeral 2}) analyzing the effect of different types of household demand and values of the cancellation fee, and ({\romannumeral 3}) justifying the necessity of developing proper online algorithms for our newly proposed problem \textbf{dSP}.

\subsection{Dataset and Parameters}

\subsubsection{Electricity Demand} The demand traces are from Office of Energy Efficiency \& Renewable Energy (EERE) \cite{openeidata}. We use the electrical dataset `EPLUS TMY3 residential', which contains three demand types of family apartments in 2013, i.e., low, base, and high with average monthly usage of {450, 900, 1430 (kWh)} respectively. Each type contains 936 unique files, each representing one year of electricity consumption for one home in a particular area in the USA \cite{smartcity}. 

\subsubsection{Energy Plans} Energy plans  in New York State \cite{powertochooseNYS} are chosen as representatives. We select four distinct energy retailers from the large number of suppliers. The four retailers, East Coast Power \& Gas, LLC, Eligo Energy NY, LLC, New York Gas \& Electric, and Renaissance Power \& Gas, Inc, all set constant cancellation fee (\$100) for a 12-month fixed-rate plan. The results here should hold similarly in other places.

\subsubsection{Electricity Prices} As only seasonal data in 2018 is shown for New York state (Zip code: 10001), we use interpolation to obtain monthly prices ($p^0_t$ and $p^1_t$). We set $H$ to be $p^0_t$ for the fixed-rate plan.

\subsubsection{Contract Period and Cancellation Fee} The fixed-rate plan has a length of 12 months and the cancellation fee of \$100. For the linearly decreasing cancellation fee plan, \$10 is charged for each month remaining in the contract when canceling the plan. These settings are consistent with most available plans in the market.

\subsubsection{Cost Benchmark} Following discussions in Section \ref{sec:intro}, baseline households do not change plans throughout the year. We set them choosing the fixed-rate plan offered by Eligo Energy NY, LLC. We evaluate their new cost by applying various algorithms.

\subsubsection{Comparisons of Algorithms} Algorithm OFA$_s$, gCHASE$_s$, gCHASE$_s^r$, and Greedy are compared under the same application scenario, where Greedy is a na\"ive online algorithm. OFA$_s$ is fed with all future input information at the beginning, whereas gCHASE$_s$, gCHASE$_s^r$, and Greedy are not given current input until the time arrives. The randomized online algorithm gCHASE$_s^r$ is judged by its average outcome over 100 times run in the same setting.

\begin{figure}[t]
\centering
  \begin{minipage}[b]{0.2\textwidth}
    	\includegraphics[width=\textwidth]{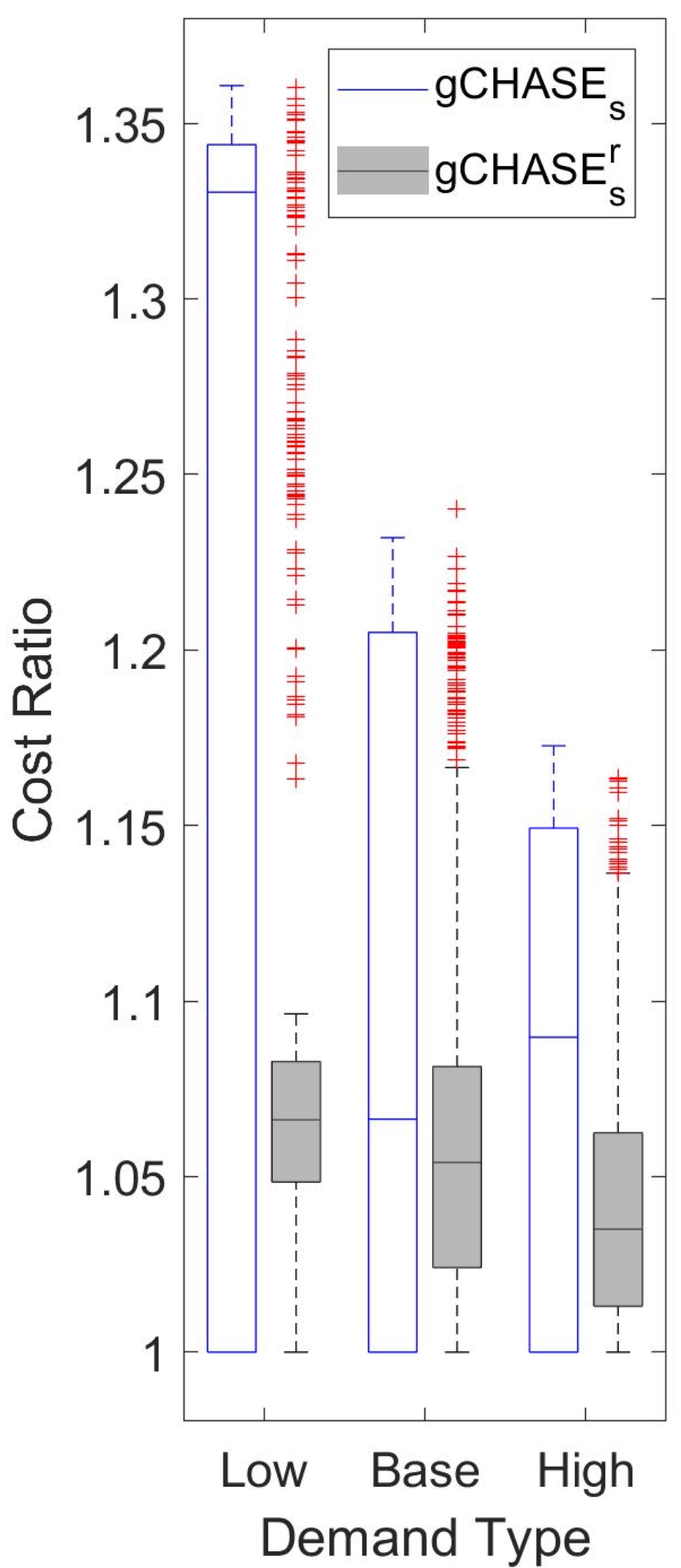}  
    	\caption{Constant Cancellation Fee -- cost ratio of online algorithms comparing to the offline algorithm under three demand types for all families.} 
    	\label{fig:const_cancel}
  \end{minipage}
  \quad
  \begin{minipage}[b]{0.205\textwidth}
    	\includegraphics[width=\textwidth]{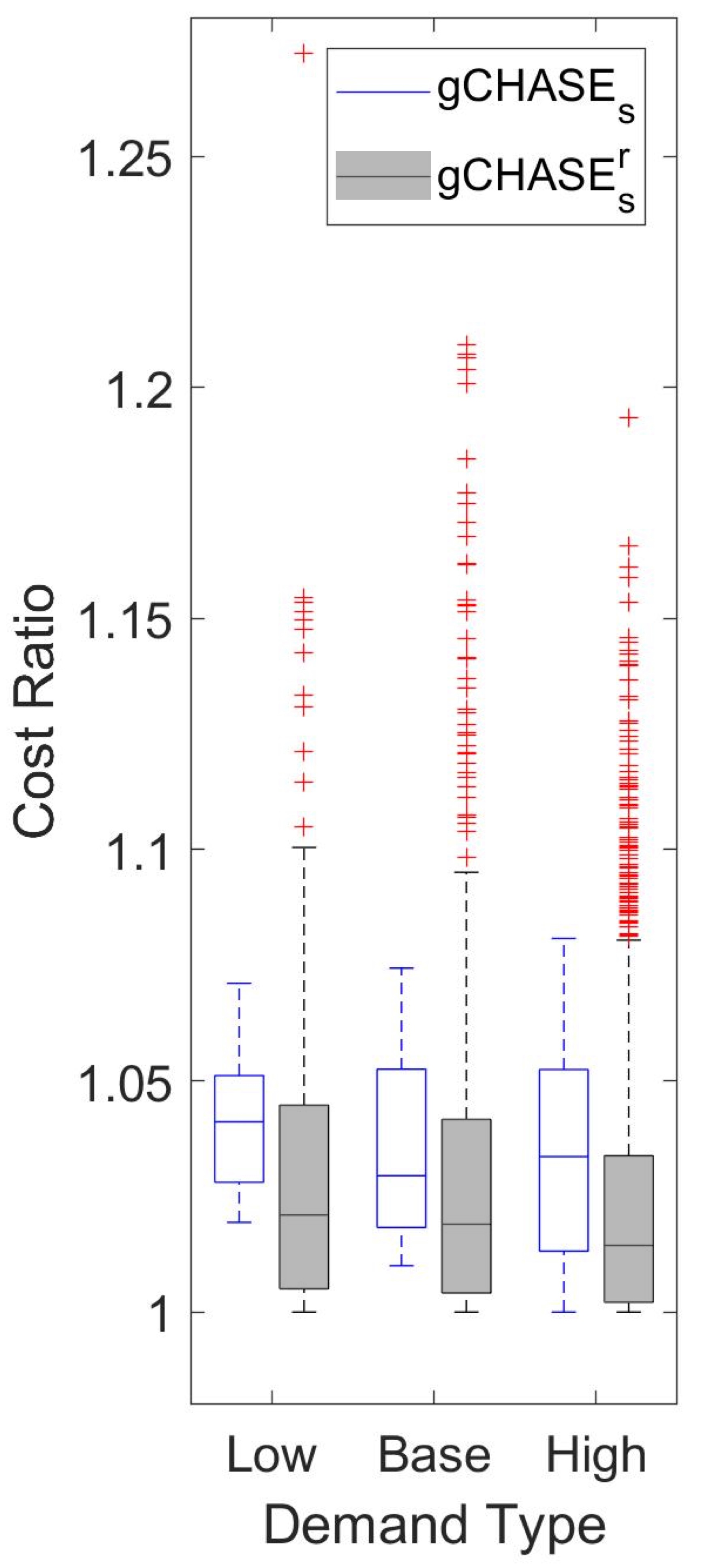}  
    	\caption{Linear Decreasing Cancellation Fee -- cost ratio of online algorithms comparing to the offline algorithm under three demand types for all families.} 
    	\label{fig:linear_cancel}
  \end{minipage}
\end{figure}

\subsection{Constant Cancellation Fee}
\label{sec:const}

\subsubsection{Purpose}
In this setting, we aim to answer two questions.
First, how well can our proposed online algorithms behave comparing to the optimal one? The offline algorithm needs full-time input before time starts, which is not practical for households to use. On the other hand, online algorithms do not require any predictions or stochastic models of future inputs, while their performances need to be validated.
Second, what are cost-saving differences between different types of families? Will the effect of our algorithms change with input mean value? Although prediction on the exact monthly energy usage is hard for an individual family, categorizing the overall average usage shall be applicable. The results aim to shed light on the relationship between the benefit of using an online algorithm and the demand-type, and give directions to families later.

\subsubsection{Observations}

From Figure \ref{fig:const_cancel}, we observe cost ratios between online algorithms and $\text{OFA}_s$ are smaller than 1.4 for all families, which is much less than theoretical competitive ratios in Theorem \ref{thm:det} and \ref{thm:gchaser}. It means that online algorithms cost no more than an additional 37\% expense compared to the optimal offline algorithm. Further, the average cost ratio of gCHASE$_s^r$ is considerably smaller than that of gCHASE$_s$. Specifically, most of the households have no more than an additional 9\% expense than the optimal offline algorithm. These results justify the importance of switching energy plans properly and introducing randomization. 

As for demand types, we note that online algorithms' behaviors get better as average demand increases. Since the cancellation fee is fixed, the impact of incurring the cancellation fee fades away when the high inevitable energy cost is due.

{\subsection{Linearly Decreasing Cancellation Fee}}

\subsubsection{Purpose} For the case when the cancellation fee is linearly decreasing with the time of enrollment in a fixed-rate plan, we implement the same online algorithms to cost data as in Section \ref{sec:const}. The optimal offline solution is derived by applying the dynamic programming technique.
In this experiment, we would like to identify the performances of our proposed online algorithms and influences of changing problem settings, aiming to find out if low cost ratios can still be maintained. 

\subsubsection{Observations}
Comparing to Figure \ref{fig:const_cancel},
Figure \ref{fig:linear_cancel} shows dramatic changes in the outcome of the deterministic online algorithm gCHASE$_s$ for all demand types.
Since the cancellation fee decreases as time pass by, the penalty of wrong decisions become relatively less. Whereas for the randomized online algorithm gCHASE$_s^r$, as it may change states frequently and bring in a high cancellation fee, it cannot guarantee its superiority comparing to gCHASE$_s$. Overall, we conclude that online algorithms have a more significant cost reduction in this section comparing to the constant cancellation fee setting, although some individual households may suffer more than applying gCHASE$_s$.

\begin{figure}[t]
\centering
    \begin{minipage}[b]{0.5\textwidth}
    	\includegraphics[width=\textwidth]{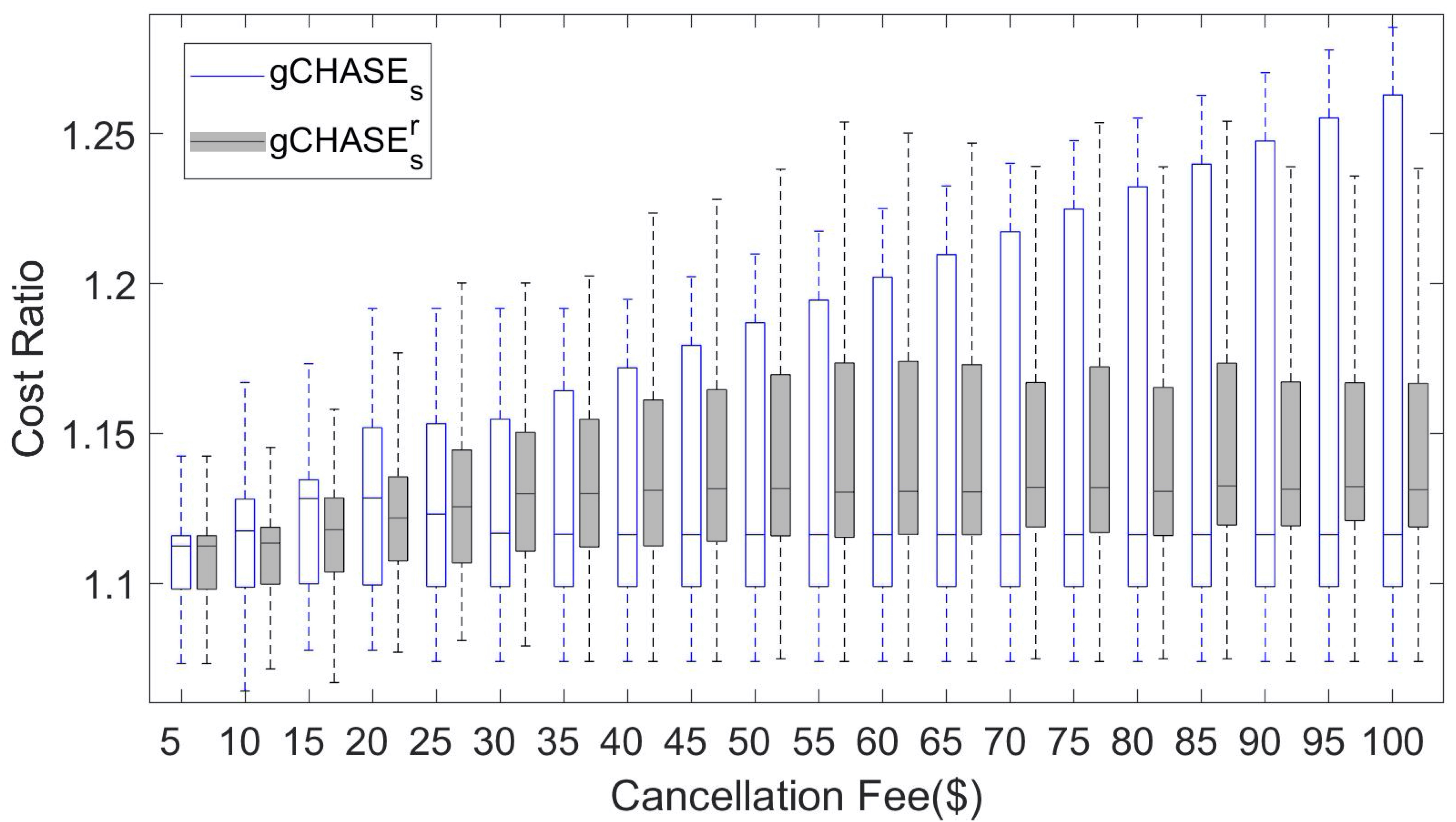}  
    	\caption{Change of Constant Cancellation Fee Under Base Demand Type -- cost ratio of online algorithms comparing to the offline algorithm for all families.} 
    	\label{fig:save_beta_box_cr}
  \end{minipage}
\end{figure}
\begin{figure}[t]
\centering
  \begin{minipage}[b]{0.5\textwidth}
    	\includegraphics[width=\textwidth]{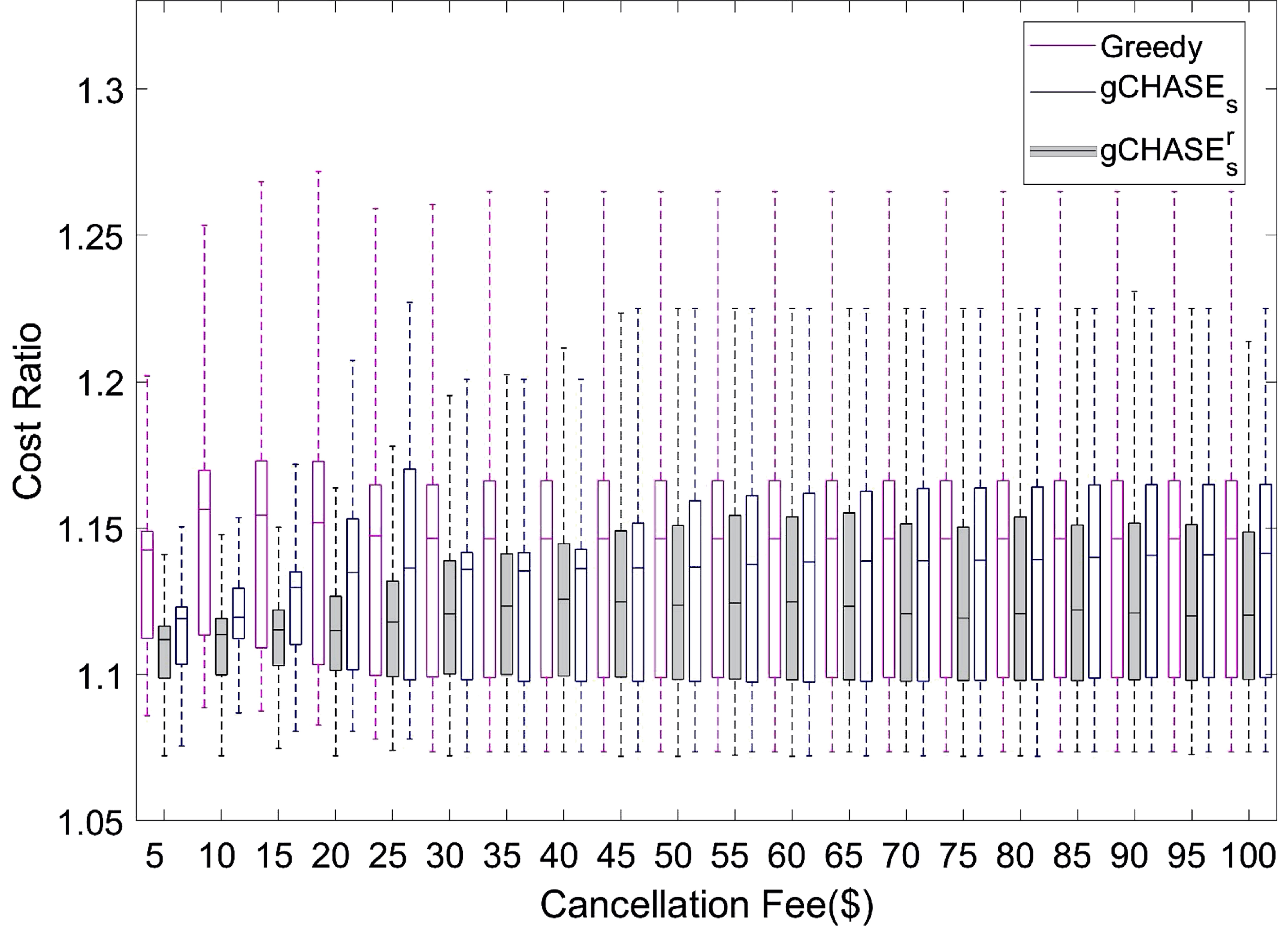}  
    	\caption{Change of Linear Decreasing Cancellation Fee Under Base Demand Type -- cost ratio of online algorithms comparing to the offline algorithm for all families.} 
    	\label{fig:save_beta_dim_box_cr}
  \end{minipage}
\end{figure}

\subsection{Effect of Cancellation Fee}
\subsubsection{Purpose} 
Previous experiments show that our proposed online algorithms guarantee small cost ratios. This section is designed to investigate the impact of changing the cancellation fee on cost ratios while keeping other parameters unchanged. Based on intuition, a high cancellation fee prevents consumers from taking the initiative in choosing to change plans due to their fear of heavy punishment. 

Mathematically, the plan length {\sffamily L} stays still, and cancellation fee $\beta$ varies from \$5 to \$100, with a \$5 increment between every two trials. For the ease of demonstration, we only show results of the `base' demand type. It is foreseeable that other demand types follow. As the lack of theoretical guarantee for problem \textbf{dSP}, we add a na\"ive online algorithm Greedy to be compared with. Figure \ref{fig:save_beta_box_cr} and \ref{fig:save_beta_dim_box_cr} plot the cost ratio change of all `base' type families against 20 distinct cancellation fee settings.

\subsubsection{Observations}
First, Figure \ref{fig:save_beta_box_cr} and \ref{fig:save_beta_dim_box_cr} generally reveal the same trends for gCHASE$_s$ and gCHASE$_s^r$. Cost ratio increases with the growth of the cancellation fee. When cancellation fee takes a large value, the impact of making a wrong decision incurs a higher penalty and its competitiveness decreases. It seems that the average cost ratio of all algorithms becomes `saturated' when the cancellation fee is high, which may imply the choice of cancellation fee for these companies in the real world.
Second, in each trial, the randomized online algorithm always performs better than the deterministic one on average, and the linear decreasing cancellation fee setting suggests smaller cost ratios. In Figure \ref{fig:save_beta_box_cr}, the smaller the competitive ratio is, the greater the dominance gCHASE$_s^r$ has. Whereas Figure \ref{fig:save_beta_dim_box_cr} indicates that this randomization may not gain much improvement in linear decreasing cancellation fee case. Third, although problem \textbf{dSP} is a new and hard challenge, our proposed online algorithms are much better than the na\"ive Greedy algorithm, which justifies the importance of developing competitive online algorithms.

\subsection{To Stay or to Switch}

After showing performances of our proposed online algorithms in constant and linearly decreasing cancellation fee settings, we provide some suggestions to households on choosing energy plans:
\begin{itemize}
    \item Households should always choose the least expensive variable-rate plan for the current month from all available ones provided by retailers.
    \item Rank fixed-rate plans by their average monthly rate. Households should better choose the one with a low rate rank, a low cancellation fee (average \$ / month) and/or a linearly decreasing setting.
    \item Households should consider implementing a randomized online algorithm.
\end{itemize}
\section{conclusion and Future Work}
\label{sec:conc}

We present effective online decision algorithms in this paper, aiming to save energy cost in a competitive energy retail market by providing the energy plan selection. Both offline optimal and competitive online algorithms are stated and proved. In the constant switching cost scenario, our deterministic and randomized online algorithms are proved to achieve the smallest possible competitive ratios in their classes. An online heuristic algorithm is developed and conjectured to perform well in a more general case with linearly decreasing switching costs.
Using real-world data traces, the effectiveness of the proposed algorithms is corroborated.
Meanwhile, empirical evaluations reveal that considerable energy cost savings can be brought to households by opportunistic switching among energy plans. 

As a further step, we intend to apply our algorithms to more real-world data traces and judge their performances. Moreover, we desire to study a class of metrical task system problems with temporally dependent switching costs. We consider it is theoretically important and practically useful to obtain competitive online algorithms for such problems. The extension should be integral in solving not only general energy plan selection problems but also other online decision problems in broad applications.



\bibliographystyle{IEEEtran}
\bibliography{ref}

\clearpage
\appendices
\appendix

\subsection{Proof of Proposition \ref{thm:sp}}
\label{pro:sp}
To show that (\ref{prob1}) and (\ref{prob2}) are equivalent, it suffices to show that 
\begin{equation}
\sum ^{T+1}_{t=1}\beta (s_{t}-s_{t-1})^+=
\sum ^{T+1}_{t=1}\frac{\beta}{2}\left|s_{t}-s_{t-1}\right|
\label{equ-provebeta}
\end{equation}
Since $s_0=s_{T+1}=0$, we obtain
\begin{equation}
\label{equ-pos}
\sum ^{T+1}_{t=1}(s_{t}-s_{t-1})=0
\end{equation}
By the definition of operation $(\cdot)^+$ and $\left|\cdot\right|$, we obtain
\begin{align}
\label{equ-cases1}
s_{t}-s_{t-1}&=
\begin{cases}
(s_{t}-s_{t-1})^+, s_{t}-s_{t-1}\geq 0\\
-(s_{t-1}-s_{t})^+, s_{t}-s_{t-1}<0
\end{cases}\\
\label{equ-cases2}
\left|s_{t}-s_{t-1}\right|&=
\begin{cases}
(s_{t}-s_{t-1})^+, s_{t}-s_{t-1}\geq 0\\
(s_{t-1}-s_{t})^+, s_{t}-s_{t-1}<0
\end{cases}
\end{align}
By substituting (\ref{equ-cases1}) into (\ref{equ-pos}), we obtain
\begin{align}
0=\sum ^{T+1}_{t=1}(s_{t}-s_{t-1})=
&\sum ^{T+1}_{t=1}(s_{t}-s_{t-1})^+\Big|_{s_{t}-s_{t-1}\geq 0}
\notag\\
&-\sum ^{T+1}_{t=1}(s_{t-1}-s_{t})^+\Big|_{s_{t}-s_{t-1}<0}
\label{equ-zero}
\end{align}
Further, we substitute (\ref{equ-cases2}) into (\ref{equ-zero}) and obtain
\begin{equation}
\sum ^{T+1}_{t=1}\left|s_{t}-s_{t-1}\right|\Big|_{s_{t}-s_{t-1}\geq 0}
=\sum ^{T+1}_{t=1}\left|s_{t}-s_{t-1}\right|\Big|_{s_{t}-s_{t-1}<0}
\end{equation}
Therefore,
\begin{align}
\sum ^{T+1}_{t=1}\left|s_{t}-s_{t-1}\right|=&
2\sum ^{T+1}_{t=1}\left|s_{t}-s_{t-1}\right|\Big|_{s_{t}-s_{t-1}\geq 0}\\
=&2\sum ^{T+1}_{t=1}(s_{t}-s_{t-1})^+\Big|_{s_{t}-s_{t-1}\geq 0}
\end{align}
By multiplying $\frac{\beta}{2}$ on both sides, we obtain (\ref{equ-provebeta}), which completes the proof.
\qed

\subsection{Proof of Theorem \ref{thm:ofa}}
\label{pro:ofa}

Let $\bm{z} \triangleq (z_1, \cdots, z_T)$ be the solution vector of OFA$_s$. Suppose $\bm{z^*} \triangleq (z^*_1, \cdots, z^*_T)$ is an optimal solution for \textbf{SP} and different from $\bm{z}$. We use contradiction to prove by showing that ${\rm Cost}(\bm{z^*}) \geq {\rm Cost}(\bm{z})$ holds for any input sequences.

Divide $[0,T+1]$ into segments:
\begin{align}
[T_0, T_1], \cdots, [T_i + 1, T_{i+1}], \cdots, [T_{2k} + 1, T_{2k+1}],
\end{align}
where $T_0 = 0$, $T_{2k+1} = T+1$, and
\begin{subequations}
\begin{numcases}{z_t = }
0, \quad t\in [T_{2i} + 1, T_{2i+1}], \quad \forall i\in [0,k]\\
1, \quad t\in [T_{2i+1} + 1, T_{2i+2}] \quad \forall i\in [0,k-1]
\end{numcases}
\end{subequations}
with $T_{2i+1} \geq T_{2i} + 1$.

Define for any $0\leq \tau_l \leq \tau_r \leq T+1$,
\begin{align}
\label{def:cost_tau}
{\rm Cost}(\bm{s_{[\tau_l,\tau_r]}}) \triangleq \sum_{t = \tau_l}^{\tau_r} \Big(g_t(s_t) + \beta \cdot (s_t - s_{t - 1})^+\Big) + \beta \cdot (s_{\tau_r + 1} - s_{\tau_r})^+
\end{align}
where $s_{-1} = s_{T+1} = 0$.

\textsl{Case-1}: $z^*_t = 1$, for any $t \in [\tau_l,\tau_r] \subseteq [T_{2i} + 1, T_{2i+1}]$

From Algorithm 1, we conclude that
\begin{subequations}
	\begin{numcases}{\Delta(t)}
	= 0, \quad t = T_{2i},\\
	\in [-\beta, 0), \quad t\in [T_{2i} + 1, T_{2i+1}],\\
	= -\beta, \quad t = T_{2i+1}.
	\end{numcases}
\end{subequations}

By (\ref{equ-delta}) and (\ref{equ-Delta}), for any $t \in [T_{2i} + 1, T_{2i+1}]$,
\begin{align}
\Delta(t)= & \max\{\Delta(t-1)+\delta(t), -\beta\}\\
\geq &  \Delta(t-1) + g_t(0) - g_t(1).
\end{align}

Then we obtain
\begin{align}
{\rm Cost} & (\bm{z^*_{[\tau_l,\tau_r]}}) - {\rm Cost}(\bm{z_{[\tau_l,\tau_r]}})\notag\\
& = \sum_{t = \tau_l}^{\tau_r}\Big(g_t(1)-g_t(0)\Big) + \beta \cdot (1-z^*_{\tau_l - 1}) - \beta \cdot z_{\tau_r + 1}\\
& \geq \sum_{t = \tau_l}^{\tau_r}\Big(\Delta(t-1) - \Delta(t)\Big)  + \beta \cdot (1 - z^*_{\tau_l - 1} - z_{\tau_r + 1}) \\
& = \Delta(\tau_l-1) - \Delta(\tau_r) + \beta \cdot (1 - z^*_{\tau_l - 1} - z_{\tau_r + 1})
\end{align}

\textsl{Case-1a}: $\tau_l = T_{2i} + 1$ and $\tau_r = T_{2i+1}$
\begin{align}
{\rm Cost}(\bm{z^*_{[\tau_l,\tau_r]}}) - & {\rm Cost}(\bm{z_{[\tau_l,\tau_r]}})\notag\\
& \geq 0 - (-\beta) + \beta \cdot (1 - z^*_{\tau_l - 1} - 1) \geq 0
\end{align}

\textsl{Case-1b}: $\tau_l > T_{2i} + 1$ and $\tau_r = T_{2i+1}$
\begin{align}
{\rm Cost}(\bm{z^*_{[\tau_l,\tau_r]}}) - & {\rm Cost}(\bm{z_{[\tau_l,\tau_r]}})\notag\\
& \geq \Delta(\tau_l-1) - (-\beta) + \beta \cdot (1 - 0 - 1)  \geq 0
\end{align}

\textsl{Case-1c}: $\tau_l = T_{2i} + 1$ and $\tau_r < T_{2i+1}$
\begin{align}
{\rm Cost}(\bm{z^*_{[\tau_l,\tau_r]}}) - & {\rm Cost}(\bm{z_{[\tau_l,\tau_r]}})\notag\\
& \geq 0 - \Delta(\tau_r) + \beta \cdot (1 - z^*_{\tau_l - 1} - 0)  \geq 0
\end{align}

\textsl{Case-1d}: $\tau_l > T_{2i} + 1$ and $\tau_r < T_{2i+1}$
\begin{align}
{\rm Cost}(\bm{z^*_{[\tau_l,\tau_r]}}) - & {\rm Cost}(\bm{z_{[\tau_l,\tau_r]}})\notag\\
& \geq \Delta(\tau_l-1) - \Delta(\tau_r) + \beta \cdot (1 - 0 - 0)  \geq 0
\end{align}

\textsl{Case-2}: $z^*_t = 0$, for any $t \in [\tau_l,\tau_r] \subseteq [T_{2i+1} + 1, T_{2i+2}]$

From Algorithm 1, we conclude that
\begin{subequations}
	\begin{numcases}{\Delta(t)}
	= -\beta, \quad t = T_{2i+1},\\
	\in (-\beta, 0], \quad t\in [T_{2i+1} + 1, T_{2i+2}],\\
	= 0, \quad t = T_{2i+2}.
	\end{numcases}
\end{subequations}

By (\ref{equ-delta}) and (\ref{equ-Delta}), for any $t \in [T_{2i+1} + 1, T_{2i+2}]$,
\begin{align}
\Delta(t)= & \min\{\Delta(t-1)+\delta(t), 0\}\\
\leq &  \Delta(t-1) + g_t(0) - g_t(1)
\end{align}

Then we obtain
\begin{align}
{\rm Cost} & (\bm{z^*_{[\tau_l,\tau_r]}}) - {\rm Cost}(\bm{z_{[\tau_l,\tau_r]}})\notag\\
& = \sum_{t = \tau_l}^{\tau_r}\Big(g_t(0)-g_t(1)\Big) + \beta \cdot z^*_{\tau_r + 1} - \beta \cdot (1-z_{\tau_l - 1})\\
& \geq \sum_{t = \tau_l}^{\tau_r}\Big(\Delta(t) - \Delta(t-1)\Big) + \beta \cdot (z^*_{\tau_r + 1} + z_{\tau_l - 1} - 1)\\
& = \Delta(\tau_r) - \Delta(\tau_l-1) + \beta \cdot (z^*_{\tau_r + 1} + z_{\tau_l - 1} - 1)
\end{align}

\textsl{Case-2a}: $\tau_l = T_{2i+1} + 1$ and $\tau_r = T_{2i+2}$
\begin{align}
{\rm Cost}(\bm{z^*_{[\tau_l,\tau_r]}}) - & {\rm Cost}(\bm{z_{[\tau_l,\tau_r]}})\notag\\
& \geq 0 - (-\beta) + \beta \cdot (z^*_{\tau_r + 1} + 0 - 1) \geq 0
\end{align}

\textsl{Case-2b}: $\tau_l > T_{2i+1} + 1$ and $\tau_r = T_{2i+2}$
\begin{align}
{\rm Cost}(\bm{z^*_{[\tau_l,\tau_r]}}) - & {\rm Cost}(\bm{z_{[\tau_l,\tau_r]}})\notag\\
& \geq 0 - \Delta(\tau_l-1) + \beta \cdot (z^*_{\tau_r + 1} + 1 - 1)  \geq 0
\end{align}

\textsl{Case-2c}: $\tau_l = T_{2i+1} + 1$ and $\tau_r < T_{2i+2}$
\begin{align}
{\rm Cost}(\bm{z^*_{[\tau_l,\tau_r]}}) - & {\rm Cost}(\bm{z_{[\tau_l,\tau_r]}})\notag\\
& \geq \Delta(\tau_r) - (-\beta) + \beta \cdot (1 + 0 - 1)  \geq 0
\end{align}

\textsl{Case-2d}: $\tau_l > T_{2i+1} + 1$ and $\tau_r < T_{2i+2}$
\begin{align}
{\rm Cost}(\bm{z^*_{[\tau_l,\tau_r]}}) - & {\rm Cost}(\bm{z_{[\tau_l,\tau_r]}})\notag\\
& \geq \Delta(\tau_r) - \Delta(\tau_l-1) + \beta \cdot (1 + 1 - 1)  \geq 0
\end{align}

In summation, as long as there exists $z_t^* \neq z_t$ for some $t$, ${\rm Cost}(\bm{z^*}) \geq {\rm Cost}(\bm{z})$ always holds. Therefore $\bm{z_t}$ is an optimal offline solution.
\qed

\subsection{Proof of Theorem \ref{thm:det}}
\label{pro:det}

Let $\bm{z} \triangleq (z_1, \cdots, z_T)$ be the solution vector of OFA$_s$ and $\bm{x} \triangleq (x_1, \cdots, x_T)$ be the solution vector of gCHASE$_s$. We prove that  ${\rm Cost}(\bm{x}) \leq 3{\rm Cost}(\bm{z})$ holds for any input sequences.

If ${\rm Cost}(\bm{z}) < \beta$, both $\bm{z}$ and $\bm{x}$ must be zero vectors, then ${\rm Cost}(\bm{x}) = {\rm Cost}(\bm{z})$. In the following, we only consider the case where ${\rm Cost}(\bm{z}) \geq \beta$.

From Algorithm 1 and Algorithm 2, we analyze time segments where $\bm{z}$ and $\bm{x}$ are different.

Define for any $1\leq \tau_l \leq \tau_r \leq T$,
\begin{align}
{\rm Cost}(\bm{s_{[\tau_l,\tau_r]}}) \triangleq \sum_{t = \tau_l}^{\tau_r} \Big(g_t(s_t) + \beta \cdot (s_t - s_{t - 1})^+\Big).
\end{align}

\textsl{Case-1:} $z_t = 0$ and $x_t = 1$, for any $t\in (\tau_l,\tau_r)$, where
\begin{subequations}
\begin{numcases}{\Delta(t)}
	= 0, \quad t = \tau_l,\\
	\in (-\beta, 0), \quad t \in (\tau_l,\tau_r),\\
	= -\beta, \quad t = \tau_r.
\end{numcases}
\end{subequations}
Let the union of such time segments $[\tau_l,\tau_r]$ be $\mathcal{T}^d$.

By (\ref{equ-delta}) and (\ref{equ-Delta}), for any $t \in (\tau_l,\tau_r)$,
\begin{align}
\Delta(t)= \Delta(t-1) + g_t(0) - g_t(1).
\end{align}

Then we obtain
\begin{equation}
\begin{split}
{\rm Cost} & (\bm{x_{[\tau_l,\tau_r]}}) - {\rm Cost}(\bm{z_{[\tau_l,\tau_r]}})\\
& = \sum_{t = \tau_l + 1}^{\tau_r - 1}\Big( g_t(1) - g_t(0) \Big) + \beta \cdot (1 - x_{\tau_l - 1}) - \beta \cdot (z_{\tau_r + 1} - 0)\\
& = \Delta(\tau_l + 1) - \Delta(\tau_r - 1) + \beta \cdot (1 - 1) - \beta \cdot (0 - 0)\\
& = \Delta(\tau_l + 1) - \Delta(\tau_r - 1) \leq \beta
\end{split}
\end{equation}

\textsl{Case-2:} $z_t = 1$ and $x_t = 0$, for any $t\in (\tau_l,\tau_r)$, where
\begin{subequations}
\begin{numcases}{\Delta(t)}
= -\beta, \quad t = \tau_l,\\
\in (-\beta, 0), \quad t \in (\tau_l,\tau_r),\\
= 0, \quad t = \tau_r.
\end{numcases}
\end{subequations}
Let the union of such time segments $[\tau_l,\tau_r]$ be $\mathcal{T}^u$.

Similarly, we obtain
\begin{equation}
\begin{split}
{\rm Cost} & (\bm{x_{[\tau_l,\tau_r]}}) - {\rm Cost}(\bm{z_{[\tau_l,\tau_r]}})\\
& = \sum_{t = \tau_l + 1}^{\tau_r - 1}\Big( g_t(0) - g_t(1) \Big) + \beta \cdot (x_{\tau_r + 1} - 0) - \beta \cdot (z_{\tau_l - 1} - 0)\\
& = \Delta(\tau_r + 1) - \Delta(\tau_l - 1) + \beta \cdot (1 - 0) - \beta \cdot (1 - 0)\\
& = \Delta(\tau_r + 1) - \Delta(\tau_l - 1) \leq \beta
\end{split}
\end{equation}

Let the rest time in which $\bm{z}$ and $\bm{x}$ are the same be $\mathcal{T}^e \triangleq \mathcal{T}\backslash \mathcal{T}^d \backslash \mathcal{T}^u$.

By summing over all time segments, we obtain
\begin{equation}
\begin{split}
& \frac{{\rm Cost}(\bm{x})}{{\rm Cost}(\bm{z})}\\
= & \frac{\sum_{\mathcal{T}^d} {\rm Cost}(\bm{x_\tau}) + \sum_{\mathcal{T}^u} {\rm Cost}(\bm{x_\tau}) + \sum_{\mathcal{T}^e} {\rm Cost}(\bm{x_\tau})}{\sum_{\mathcal{T}^d} {\rm Cost}(\bm{z_\tau}) + \sum_{\mathcal{T}^u} {\rm Cost}(\bm{z_\tau}) + \sum_{\mathcal{T}^e} {\rm Cost}(\bm{z_\tau})}\\
\leq & \frac{\beta + \sum_{\mathcal{T}^d} {\rm Cost}(\bm{z_\tau}) + \beta + \sum_{\mathcal{T}^u} {\rm Cost}(\bm{z_\tau}) + \sum_{\mathcal{T}^e} {\rm Cost}(\bm{z_\tau})}{\sum_{\mathcal{T}^d} {\rm Cost}(\bm{z_\tau}) + \sum_{\mathcal{T}^u} {\rm Cost}(\bm{z_\tau}) + \sum_{\mathcal{T}^e} {\rm Cost}(\bm{z_\tau})}\\
= & 1 + 2 \frac{\beta}{{\rm Cost}(\bm{z})} \leq 1 + 2\frac{\beta}{\beta} = 3
\end{split}
\end{equation}

\qed

\subsection{Proof of Lemma \ref{lem-ran}}
\label{pro:lem-ran}

According to the definition in (\ref{equ:sp}), we have
\begin{equation*}
\left\{
\begin{aligned}
    {\mathbb E}[{\rm Cost}(\bm{s})] & = \sum_{t=1}^{T}\Big({\mathbb E}[g_t(s_t)] + \beta\mathbb{E}[(s_t-s_{t-1})^+]\Big)\\
    {\rm Cost}({\mathbb E}[\bm{s}]) & = \sum_{t=1}^{T}\Big(g_t({\mathbb E}[s_t]) + \beta(\mathbb{E}[s_t]-\mathbb{E}[s_{t-1}])^+\Big).
\end{aligned}
\right.
\end{equation*}

For the ease of notation, we denote $\mathbb{E}[s_t]$ by $\overline{s_t}$. To show ${\mathbb E}[{\rm Cost}(\bm{s})] = {\rm Cost}({\mathbb E}[\bm{s}])$, it is sufficient to prove the following equations hold:
\begin{subnumcases}{}
\mathbb{E}[g_t(s_t)]  = g_t(\overline{s_t}),  \forall t\in [1,T] \label{equ:lem-1}\\
\mathbb{E}[(s_t - s_{t-1})^+]  = (\overline{s_t}-\overline{s_{t-1}})^+,  \forall t\in [1,T] \label{equ:lem-2}
\end{subnumcases}

In problem \textbf{SP}, $s_t$ only takes value 0 or 1, therefore $\overline{s_t} = \text{Pr}[s_t = 1]$. We linear interpolate between $g_t(0)$ and $g_t(1)$ and get 
\begin{equation}
    g_t(x) \triangleq g_t(0) + \Big(g_t(1) - g_t(0)\Big)\cdot x, \forall x\in [0,1],
\end{equation}

Therefore, in (\ref{equ:lem-1}) we have
    \begin{align}
    \mathbb{E}[g_t(s_t)]
    & = g_t(0)\cdot\text{Pr}[s_t = 0] + g_t(1)\cdot\text{Pr}[s_t = 1] \notag \\
    & = g_t(0)(1-\text{Pr}[s_t = 1]) + g_t(1) \text{Pr}[s_t = 1]\notag\\
    & = g_t(0) + (g_t(1) -g_t(0))\text{Pr}[s_t = 1]\notag\\
    & = g_t(\text{Pr}[s_t = 1])\\
    & = g_t(\overline{s_t})
    \end{align}
    
Note that for arbitrary integer set $\mathcal{N}\triangleq\{0,1,\cdots,n\}$, if function $g_t(\cdot)$ is linear after linear interpolation between each two consecutive integers $i$ and $i+1$, where $i\in\mathcal{N}\setminus\{n\}$, i.e.,
\begin{equation}
    g_t(i+x) \triangleq g_t(i) + \Big(g_t(i+1) - g_t(i)\Big)\cdot x, \forall x\in [0,1],
\end{equation}
then (\ref{equ:lem-1}) still holds by the linear property of expectation.

Next, we show (\ref{equ:lem-2}) holds for arbitrary $\mathcal{N}$. By definition of expectation, $\overline{s_t} = \sum_{i=0}^{n}i\cdot\text{Pr}[s_t = i]$. Denote $\text{Pr}[s_{t-1} = i, s_{t} = j]$ by $p_t(i,j)$.

\textsl{Case-1: } $\text{Pr}[s_t > s_{t-1}] = 0$, i.e., $p_t(i,j) = 0, \forall i < j$.

It is straightforward that $\mathbb{E}\Big((s_t - s_{t-1})^+\Big) = 0$.
\begin{align}
\overline{s_t} - \overline{s_{t-1}}
& =
\sum_{j=0}^n \Big(j\cdot\sum_{i=0}^n p_t(i,j))- \sum_{i=0}^n (i\cdot\sum_{j=0}^n p_t(i,j)\Big)\\
& = 
\sum_{j=0}^{i-1} \Big(j\cdot\sum_{i=1}^n p_t(i,j)\Big)- \sum_{i=j+1}^{n} \Big(i\cdot\sum_{j=0}^{n-1} p_t(i,j)\Big)\\
& = 
\sum_{j=0}^{i-1} \sum_{i=1}^n \Big(j\cdot p_t(i,j)\Big)- \sum_{i=j+1}^{n} \sum_{j=0}^{n-1} \Big(i\cdot p_t(i,j)\Big)\\
& = 
\sum_{i=j+1}^{n} \sum_{j=0}^{n-1} \Big(j\cdot p_t(i,j)\Big)- \sum_{i=j+1}^{n} \sum_{j=0}^{n-1} \Big(i\cdot p_t(i,j)\Big)\\
& = 
\sum_{i=j+1}^{n} \sum_{j=0}^{n-1} \Big((j - i)\cdot p_t(i,j)\Big)\\
& < 0
\end{align}
Then we have $(\overline{s_t} - \overline{s_{t-1}})^+ = 0$.

\textsl{Case-2: } $\text{Pr}[s_t < s_{t-1}] = 0$, i.e., $p_t(i,j) = 0, \forall i > j$.
\begin{equation}
    \mathbb{E}\Big((s_t-s_{t-1})^+\Big) = \sum_{j=i+1}^n\sum_{i=0}^n(j-i)\cdot p_{t}(i,j)
\end{equation}
\begin{align}
\overline{s_t} - \overline{s_{t-1}}
& =
\sum_{j=0}^n \Big(j\cdot\sum_{i=0}^n p_t(i,j)\Big)- \sum_{i=0}^n \Big(i\cdot\sum_{j=0}^n p_t(i,j)\Big)\\
& = 
\sum_{j=i+1}^n \Big(j\cdot\sum_{i=0}^{n} p_t(i,j)\Big)- \sum_{i=0}^{j-1} \Big(i\cdot\sum_{j=1}^{n} p_t(i,j)\Big)\\
& = 
\sum_{j=i+1}^n \sum_{i=0}^{n} \Big(j\cdot p_t(i,j)\Big)- \sum_{i=0}^{j-1} \sum_{j=1}^{n} \Big(i\cdot p_t(i,j)\Big)\\
& = 
\sum_{j=i+1}^n \sum_{i=0}^{n} \Big(j\cdot p_t(i,j)\Big)- \sum_{j=i+1}^n \sum_{i=0}^{n} \Big(i\cdot p_t(i,j)\Big)\\
& = 
\sum_{j=i+1}^n \sum_{i=0}^{n} \Big((j - i)\cdot p_t(i,j)\Big)\\
& > 0
\end{align}
Then we have $\mathbb{E}[(s_t - s_{t-1})^+]  = (\overline{s_t}-\overline{s_{t-1}})^+$.
    
Therefore, ${\mathbb E}[{\rm Cost}(\bm{s})] = {\rm Cost}({\mathbb E}[\bm{s}])$ holds.

\subsection{Proof of rCHASE satisfying equation (\ref{equ-lem})}
\label{pro:chase}
According to the description of algorithm rCHASE, at each time  $t$, $\gamma_{\text{on}}\cdot(\gamma_{\text{off}} + \beta) = 0$ holds. In each critical segment begin with $\Delta(t) = -\beta$, i.e., type-0 and type-1, the algorithm only decides when to change $s_t$ from 0 to 1. Once the state is changed, the state during the rest time of the interval remains stable. We have $\text{Pr}[s_t < s_{t-1}] = 0$ in these segments. Similarly, in each critical segment begin with $\Delta(t) = -\beta$, i.e., type-0 and type-1, as the algorithm only decides when to change $s_t$ from 1 to 0, we have $\text{Pr}[s_t > s_{t-1}] = 0$.

\subsection{Proof of Theorem \ref{thm:gchaser}}
\label{pro:ran}

We extend problem \textbf{SP} to the continuous problem \textbf{cSP} defined as follows:
\begin{align}
\text{\textbf{cSP:} minimize \ } & {\rm Cost}(\bm{\sigma}) \triangleq \sum_{t=1}^{T}\Big(g_t(\sigma_t,x_t)+\beta\cdot (x_t-x_{t-1})^+\Big)\\
\text{subject to \ } & x_t\in [0, 1] 
\end{align}
where
\begin{align}
	g(\sigma_t,x) = \big(g(\sigma_t,1)-g(\sigma_t,0)\big) \cdot x + g(\sigma_t,0) , \forall x\in [0, 1]
\end{align}

And we define an online algorithm for problem \textbf{cSP}, as shown in {Algorithm \ref{alg-cchase}}.

\setcounter{algorithm}{3}    

\begin{algorithm}[h]
	\caption{$\text{cCHASE}_{s}^{r}$}
	\begin{algorithmic}
		\label{alg-cchase}
		\STATE Set $x_0 \leftarrow 0$
		\FOR{$t$ from 1 to $T$}
		\STATE Compute $\Delta(t)$
		\STATE $x_t \leftarrow \frac{\beta + \Delta(t)}{\beta}$
		\RETURN $x_t$
		\ENDFOR
	\end{algorithmic}
\end{algorithm}

\begin{lem}
	The expected cost of algorithm $\text{gCHASE}_s^r$ is equal to the cost of algorithm $\text{cCHASE}_{s}^{r}$ given the same input.
\end{lem}

\begin{proof}
	For algorithm $\text{gCHASE}_s^r$, let $p_t \triangleq \text{Pr}[s_t = 1]$. Therefore, $\text{Pr}[s_t = 0] = 1 - p_t$.
	
	\textsl{Step-1: } show that $p_t = x_t$. We prove by induction.
	
	For $t=0$, $p_t = x_t = 0$. If $p_{t-1} = x_{t-1}$ holds,
	
	\textsl{Case-1: } $\Delta(t-1) \leq \Delta(t)$
	\begin{align}
		\text{Pr}[s_t = 1] & = \text{Pr}[s_t = 1|s_{t-1} = 1] \cdot \text{Pr}[s_{t-1} = 1]\\
		& + \text{Pr}[s_t = 1|s_{t-1} = 0]\cdot \text{Pr}[s_{t-1} = 0]\\
		& = 1 \cdot p_{t-1} + (1-\frac{\Delta(t)}{\Delta(t-1)}) \cdot (1-p_{t-1})\\
		& = x_{t-1} + (1-\frac{\beta \cdot x_{t} - \beta}{\beta \cdot x_{t-1} - \beta}) \cdot (1-x_{t-1})\\
		& = x_{t}
	\end{align}
	
	\textsl{Case-2: } $\Delta(t-1) > \Delta(t)$
	\begin{align}
	\text{Pr}[s_t = 1] & = \text{Pr}[s_t = 1|s_{t-1} = 1] \cdot \text{Pr}[s_{t-1} = 1]\\
	& + \text{Pr}[s_t = 1|s_{t-1} = 0]\cdot \text{Pr}[s_{t-1} = 0]\\
	& = \frac{\beta + \Delta(t)}{\beta + \Delta(t-1)} \cdot p_{t-1} + 0\\
	& = \frac{\beta \cdot x_{t}}{\beta \cdot x_{t-1}} \cdot x_{t-1}\\
	& = x_{t}
	\end{align}
	
	\textsl{Step-2: } show that $\mathbb{E}[g_t(\sigma_t, s_t)] = g_t(\sigma_t, x_t)$.
	\begin{align}
		\mathbb{E}[g_t(\sigma_t, s_t)] & = p_t \cdot g_t(\sigma_t, 1) + (1-p_t) \cdot g_t(\sigma_t, 0)\\
		& = x_t \cdot g_t(\sigma_t, 1) + (1-x_t) \cdot g_t(\sigma_t, 0)\\
		& = g_t(\sigma_t, x_t)
	\end{align}
	
	\textsl{Step-3: } show that $\mathbb{E}[(s_t - s_{t-1})^+] = (x_t-x_{t-1})^+$.
	\begin{align}
		\mathbb{E}[(s_t - s_{t-1})^+] & = (p_t - p_{t-1})^+\\
		& = (x_t-x_{t-1})^+
	\end{align}
	
	By summing the above over $t\ \in \mathcal{T}$, we obtain ${\rm Cost}_{\text{gCHASE}_s^r}(\bm{\sigma}) = {\rm Cost}_{\text{cCHASE}_s^r}(\bm{\sigma})$.
\end{proof} 

\begin{lem}
	Algorithm $\text{cCHASE}_{s}^{r}$ is 2-competitive for problem \textbf{cSP}.
\end{lem}

\begin{proof}
	We denote $\text{cOPT}_{s}$ as the optimal offline algorithm for problem \textbf{cSP}, and the state it chooses is $z_t$ at time $t$. We make the following simplifications and assumptions: define 
	\begin{align}
	f_t(x) \triangleq g_t(\sigma_t, x) - \min\Big\{g_t(\sigma_t, 0) , g_t(\sigma_t, 1)\Big\} , \forall t \in \mathcal{T}
	\end{align}
	and category them to two kinds of functions, namely
	\begin{subequations}
	\begin{numcases}{f^<_t(x) = }
		0, & $x = 0$\\
		x \cdot f_t(1), & $x \in (0,1)$\\
		f_t(1) \geq 0, & $x = 1$
	\end{numcases}
	\end{subequations}
	and
	\begin{subequations}
	\begin{numcases}{f^>_t(x) = }
	f_t(0) > 0, & $x = 0$\\
	(1-x) \cdot f_t(0), & $x \in (0,1)$\\
	0, & $x = 1$
	\end{numcases}
	\end{subequations}
	for any $t\in \mathcal{T}$. 
	
	Further, by combining definitions of $\Delta(t)$, $f_t(x)$ and $x_t$, we observe that
	\begin{align}
	x_t = \Big(x_{t-1} + \frac{f_t(0)-f_t(1)}{\beta}\Big)^1_0 , \forall t \in \mathcal{T}. \label{equ:xt}
	\end{align}
	
	Therefore,
	\begin{subequations}
	\begin{numcases}{}
	x_t \leq x_{t-1}, \quad f_t(x) = f^<_t(x) \label{equ:x<}\\
	x_t \geq x_{t-1}, \quad f_t(x) = f^>_t(x) \label{equ:x>}
	\end{numcases}
	\end{subequations}
		
	We use the potential function argument which needs to show the following
	\begin{subequations}
	\begin{numcases}{}
	f_t(x_t) + \beta \cdot (x_t-x_{t-1})^+ +  \Phi(x_t,z_t) - \Phi(x_{t-1},z_{t-1}) \notag\\
	\qquad \qquad \leq 2 \cdot \Big( f_t(z_t) + \beta \cdot (z_t-z_{t-1})^+ \Big) , \forall t \in \mathcal{T} \label{equ:x2z}\\
    \Phi(x_t,z_t) \geq 0 , \forall t \in \mathcal{T} \label{equ:phi}
	\end{numcases}
	\end{subequations}
    where    
    \begin{align}
    \Phi(x_t,z_t) \triangleq \beta \cdot (\frac{1}{2}x_t^2 + 2z_t - 2z_t\cdot x_t) , \forall t \in \mathcal{T}.
    \end{align}
 
    \textsl{Step-1: } show that (\ref{equ:phi}) holds. 
    \begin{align}
    \Phi(x_t,z_t) = \beta \cdot \Big[2z_t(1-x_t)+\frac{1}{2}x_t^2\Big]\geq 0, \forall t \in \mathcal{T}.
    \end{align}
    
    \textsl{Step-2: } show that (\ref{equ:x2z}) holds.
    
    We use the idea of interleaving moves \cite{BEY05online}. Suppose at each time $t$, only $\text{cOPT}_{s}$ or $\text{cCHASE}_{s}^{r}$ changes its state. If both of them change their states, the input function can be decomposed into two stages, and feed them one after another so that two algorithms change states sequentially. 
    
    \textsl{Case-1: $x_t = x_{t-1}$, $z_t \neq z_{t-1}$}.
    
    \textsl{Case-1a: $f_t(x_t) = f^<_t(x_t)$ }. By (\ref{equ:xt}) and (\ref{equ:x<}), $x_t = 0$.
    \begin{align}
    LHS & = f_t(0) + \beta \cdot (0-0)^+ +  \Phi(0,z_t) - \Phi(0,z_{t-1})\\
    & = 0 + 0 + \beta \cdot (2z_t) - \beta \cdot (2z_{t-1})\\
    & = 2\beta (z_t - z_{t-1}) \leq RHS    
    \end{align}
    
    \textsl{Case-1b: $f_t(x_t) = f^>_t(x_t)$}. By (\ref{equ:xt}) and (\ref{equ:x>}), $x_t = 1$.
    \begin{align}
    LHS & = f_t(1) + \beta \cdot (1-1)^+ +  \Phi(1,z_t) - \Phi(1,z_{t-1})\\
    & = 0 + 0 + \beta \cdot (2z_t - 2z_t) - \beta \cdot (2z_{t-1} - 2z_{t-1})\\
    & = 0 \leq RHS    
    \end{align}
    
    \textsl{Case-2: $z_t = z_{t-1}$, $x_t \neq x_{t-1}$}.
    
    \textsl{Case-2a: $f_t(x_t) = f^<_t(x_t)$}.
    
    \textsl{Case-2a$^{\uppercase\expandafter{\romannumeral1}}$: $x_t \neq 0$}. By (\ref{equ:xt}), $x_t = x_{t-1} - \frac{f_t(1)}{\beta}$.
    \begin{align}
    LHS - RHS = & f_t(x_t) + 0 + \Phi(x_t,z_t) - \Phi(x_{t-1},z_t) - 2f_t(z_t)\\
    = & x_t\cdot f_t(1) + \beta \cdot (\frac{1}{2}x_t^2 - 2z_t\cdot x_t) \notag\\
    & - \beta \cdot (\frac{1}{2}x_{t-1}^2 - 2z_t\cdot x_{t-1}) - 2 z_t\cdot f_t(1)\\
    = & - \frac{1}{2}\beta\cdot (x_t - x_{t-1})^2 \leq 0
    \end{align}
    
    \textsl{Case-2a$^{\uppercase\expandafter{\romannumeral2}}$: $x_t = 0$}. By (\ref{equ:xt}), $x_{t-1} -\frac{f_t(1)}{\beta} \leq 0$.
    \begin{align}
    LHS - RHS = & f_t(0) + 0 + \Phi(0,z_t) - \Phi(x_{t-1},z_t) - 2f_t(z_t)\\
    = & 0 - \beta \cdot (\frac{1}{2}x_{t-1}^2 - 2z_t\cdot x_{t-1}) - 2 z_t\cdot f_t(1)\\
    \leq & - \frac{1}{2}\beta\cdot x_{t-1}^2 + 2\beta\cdot z_t\cdot x_{t-1} - 2\beta \cdot z_t\cdot  x_{t-1} \\
    = & - \frac{1}{2}\beta\cdot x_{t-1}^2 \leq 0
    \end{align}
    
    \textsl{Case-2b: $f_t(x_t) = f^>_t(x_t)$}.
    
    \textsl{Case-2b$^{\uppercase\expandafter{\romannumeral1}}$: $x_t \neq 1$}. By (\ref{equ:xt}), $x_t = x_{t-1} + \frac{f_t(0)}{\beta}$.
    \begin{align}
    LHS - RHS = & f_t(x_t) + \beta\cdot\frac{f_t(0)}{\beta} + \Phi(x_t,z_t) \notag\\
    & - \Phi(x_{t-1},z_t) - 2f_t(z_t)\\
    = & (1 - x_t)\cdot f_t(0) + f_t(0) +\beta \cdot (\frac{1}{2}x_t^2 - 2z_t\cdot x_t) \notag\\
    & - \beta \cdot (\frac{1}{2}x_{t-1}^2 - 2z_t\cdot x_{t-1}) - 2 (1 - z_t)\cdot f_t(0)\\
    = & - \frac{1}{2}\beta\cdot (x_t - x_{t-1})^2< 0
    \end{align}
    
    \textsl{Case-2b$^{\uppercase\expandafter{\romannumeral2}}$: $x_t = 1$}. By (\ref{equ:xt}), $x_{t-1}   + \frac{f_t(0)}{\beta} \geq 1$.
    \begin{align}
    LHS - RHS = & f_t(1) + \beta \cdot (1-x_{t-1}) + \Phi(1,z_t) \notag\\
    & - \Phi(x_{t-1},z_t) - 2f_t(z_t)\\
    = & 0 + \beta \cdot (1-x_{t-1}) + \beta \cdot (\frac{1}{2} - 2z_{t-1}) \notag\\
    & - \beta \cdot (\frac{1}{2}x_{t-1}^2 - 2z_t\cdot x_{t-1}) - 2 (1 - z_t)\cdot f_t(0)\\
    \leq & \beta \cdot (\frac{3}{2} - x_{t-1} - 2z_{t} - \frac{1}{2}x_{t-1}^2 + 2z_t\cdot x_{t-1}) \notag\\
    & - 2 (1 - z_t)\cdot (\beta + \beta \cdot x_{t-1})\\
    = & - \frac{1}{2}\beta\cdot (x_{t-1} - 1)^2 \leq 0
    \end{align}
    
    To show the competitive ratio of $\text{cCHASE}_{s}^{r}$ is not smaller than ($2-\epsilon$), for any $\epsilon > 0$, we construct a special input sequence:
    \begin{subequations}
    \begin{numcases}{}
    	g_1(\sigma_1,0) = \delta\beta , g_1(\sigma_1,1) = 0\\
    	g_t(\sigma_1,1) = \delta\beta , g_t(\sigma_1,0) = 0, \forall t \in \mathcal{T} \backslash \{1\}
    \end{numcases}
    \end{subequations}
    where $ \delta $ is a positive real constant smaller than 1 and $\delta \rightarrow 0$.
    
    The state vector of $\text{cCHASE}_{s}^{r}$ is:
    \begin{subequations}
    \begin{numcases}{x_t = }
    \delta , t = 1\\
    0 , \forall t \in \mathcal{T} \backslash \{1\}
    \end{numcases}
    \end{subequations}
    
    The cost of $\text{cCHASE}_{s}^{r}$ is:
    \begin{align}
    	{\rm Cost}_{\text{cCHASE}_{s}^{r}} & = \beta \cdot \delta + (1-\delta)\cdot \delta\beta + 0\\
    	& = 2\delta\beta - \delta^2 \beta
    \end{align}
    
    We claim that the state vector of $\text{cOPT}_{s}$ is $z_t = 0$ for all $t \in \mathcal{T}$, and its cost is $\delta\beta$. In this case, the cost ratio between the two algorithms is:
    \begin{align}
    \frac{{\rm Cost}_{\text{cCHASE}_{s}^{r}}}{{\rm Cost}_{\text{cOPT}_{s}}} & = \frac{2\delta\beta - \delta^2 \beta}{\delta\beta}\\
    & = 2 - \delta \rightarrow 2 - 0 = 2
    \end{align}
\end{proof}

Last but not least, we need to verify that ${\rm Cost}_{\text{cOPT}_{s}} \leq {\rm Cost}_{\text{OFA}_s}$ always holds. As in the decision space in continuous version include the one in discrete space. Therefore ${\rm Cost}_{\text{OFA}_s}$ must be a lower bound to the continuous setting.

In conclusion, we show that
\begin{align}
    \mathbb{E}[{\rm Cost}_{\text{gCHASE}_s^r}] = {\rm Cost}_{\text{cCHASE}_{s}^{r}} \leq 2{\rm Cost}_{\text{cOPT}_{s}} \leq 2{\rm Cost}_{\text{OFA}_s}
\end{align}
which completes the proof.
\qed

\subsection{Discussion on Linearly Decreasing Cancellation Fee}
\label{pro:conj}
We first show that problem \textbf{SP} can be reformulated in a similar way as \textbf{dSP}, then claim two problems only differs from a constant $\alpha$. 

Define
\begin{align}
\phi(t) \triangleq \sum_{\tau = 0}^{t - 1} \delta(\tau) 
\end{align}

By using time segments defined by (\ref{equ:segment}), we obtain
\begin{align}
	& \sum_{t=1}^{T}\Big(g_t(s_t)+\beta\cdot (s_t-s_{t-1})^+\Big)\\
	= & \sum_{i = 0}^{n} \sum_{t = T_{2i}}^{T_{2i + 1} - 1} g_t(0) + \sum_{i = 1}^{n} \sum_{t = T_{2i - 1}}^{T_{2i} - 1} g_t(1) + \beta \cdot n\\
	= & \sum_{t=1}^{T} g_t(1) + \sum_{i = 0}^{n} \sum_{t = T_{2i}}^{T_{2i + 1} - 1} \delta_t(0) + \beta \cdot n\\
	= & \sum_{t=1}^{T} g_t(1) - \beta + \sum_{i = 0}^{n} \Big( \phi(T_{2i + 1} - 1) - \phi(T_{2i} - 1) + \beta \Big)
\end{align}

Similarly, for problem \textbf{dSP}, we define
\begin{align}
\Phi(t) \triangleq \sum_{\tau = 0}^{t - 1} (\delta(\tau) - \alpha)
\end{align}

\begin{align}
& \sum_{t=1}^{T} g_t(s_t) + \sum_{i = 0}^{n}\Big(\alpha\cdot [\text{\sffamily L} - (T_{2i + 1} - T_{2i})]\Big)\\
= & \sum_{t=1}^{T} g_t(1) + \sum_{i = 0}^{n} \sum_{t = T_{2i}}^{T_{2i + 1} - 1} \delta_t(0) + \beta \cdot n - \alpha \cdot \sum_{i = 0}^{n}(T_{2i + 1} - T_{2i})\\
= & \sum_{t=1}^{T} g_t(1) - \beta + \sum_{i = 0}^{n} \Big( \Phi(T_{2i + 1} - 1) - \Phi(T_{2i} - 1) + \beta \Big)
\end{align}

By regarding $\delta(\tau) - \alpha$ as new input, we claim that \textbf{SP} differs from \textbf{dSP} by only a constant. Therefore, the previous algorithms might apply in this case.

\end{document}